\documentclass[11pt,letterpaper]{article}
\usepackage[margin=1in]{geometry}
\usepackage{amssymb,amsmath,amsthm}
\usepackage{graphicx}
\usepackage[colorlinks,citecolor=blue,linkcolor=blue,urlcolor=red,pagebackref,hypertexnames=false]{hyperref}
\usepackage{subcaption}
\usepackage{booktabs}
\usepackage{lineno}
\usepackage{color}
\usepackage{enumitem}
\usepackage{tikz}
\usepackage{algorithm}
\usepackage[noend]{algpseudocode}
\usepackage{framed}
\usepackage{thm-restate}
\usepackage[T1]{fontenc}
\usepackage[colorinlistoftodos]{todonotes}
\usepackage{comment}

\newcommand*\patchAmsMathEnvironmentForLineno[1]{%
  \expandafter\let\csname old#1\expandafter\endcsname\csname #1\endcsname
  \expandafter\let\csname oldend#1\expandafter\endcsname\csname end#1\endcsname
  \renewenvironment{#1}%
     {\linenomath\csname old#1\endcsname}%
     {\csname oldend#1\endcsname\endlinenomath}}%
\newcommand*\patchBothAmsMathEnvironmentsForLineno[1]{%
  \patchAmsMathEnvironmentForLineno{#1}%
  \patchAmsMathEnvironmentForLineno{#1*}}%
\AtBeginDocument{%
\patchBothAmsMathEnvironmentsForLineno{equation}%
\patchBothAmsMathEnvironmentsForLineno{align}%
\patchBothAmsMathEnvironmentsForLineno{flalign}%
\patchBothAmsMathEnvironmentsForLineno{alignat}%
\patchBothAmsMathEnvironmentsForLineno{gather}%
\patchBothAmsMathEnvironmentsForLineno{multline}%
}

\newtheorem{theorem}{Theorem}[section]
\newtheorem{definition}[theorem]{Definition}
\newtheorem{claim}[theorem]{\bfseries{Claim}}
\newtheorem{lemma}[theorem]{\bfseries{Lemma}}

\newtheorem*{remark}{Remark}%

\newcommand{\dm}{\mathrm{dm}}

\newcommand{\eps}{\varepsilon}

\renewcommand{\leq}{\leqslant}
\renewcommand{\geq}{\geqslant}

\newbox\ProofSym
\setbox\ProofSym=\hbox{%
\unitlength=0.18ex%
\begin{picture}(10,10)
\put(0,0){\framebox(9,9){}}
\put(0,3){\framebox(6,6){}}
\end{picture}}

\begin{document}

\title{
    Approximating Partition in Deterministic Near-Linear Time
}

\author{
  Lin Chen\thanks{chenlin198662@zju.edu.cn. Zhejiang University. }
  \and 
  Jiayi Lian\thanks{jiayilian@zju.edu.cn. Zhejiang University.}
  \and
  Yuchen Mao\thanks{maoyc@zju.edu.cn. Zhejiang University. }
  \and
  Guochuan Zhang\thanks{zgc@zju.edu.cn. Zhejiang University. }
}

\date{}
  
\maketitle
\begin{abstract}
    We propose a deterministic $\widetilde{O}(n + \frac{1}{\eps})$-time FPTAS (Fully Polynomial-Time Approximation Scheme) for the classical  Partition problem. This is the best possible (up to a polylogarithmic factor) assuming SETH (Strong Exponential Time Hypothesis) [Abboud, Bringmann, Hermelin, and Shabtay'22].  Prior to our work, the best known FPTAS for Partition runs in $\widetilde{O}(n + (\frac{1}{\eps})^{5/4})$ time [Deng, Jin and Mao'23, Wu and Chen'22].  Our result is obtained by solving a more general problem of weakly approximating Subset Sum.
\end{abstract}

\section{Introduction}

\paragraph{Subset Sum} Given a (multi-)set $X$ of $n$ positive integers and a target $t$, Subset Sum asks for a subset $Y \subseteq X$ with maximum $\Sigma(Y)$ that does not exceed $t$, where $\Sigma(Y)$ is the sum of the integers in $Y$. It is a fundamental problem in computer science and operations research. Although NP-Hard, it admits FPTASes (Fully Polynomial-Time Approximation Schemes).  The first FPTAS for Subset Sum was due to Ibarra and Kim~\cite{IK75} and Karp~\cite{Kar75} in the 1970s.  Since then, a great effort has been devoted to developing faster FPTASes (see Table~\ref{table:SubsetSum}). The latest one is by Kellerer, Pferschy and Speranza~\cite{KPS97}, and has an $\widetilde{O}(n+(\frac{1}{\eps})^2)$ running time.\footnote{In the paper we use an $\widetilde{O}(\cdot)$ notation to hide polylogarithmic factors in $n$ and $\frac{1}{\eps}$.} This was recently shown, by Bringmann and Nakos~\cite{BN21b}, to be the best possible (up to a polylogarithmic factor) assuming the $(\min,+)$-convolution conjecture.

\paragraph{Partition} Partition is a special case of Subset Sum where $t=\Sigma(X)/2$. It is often considered as one of ``the easiest NP-hard Problem'' and yet has many applications in scheduling \cite{CL91}, minimization of circuit sizes and cryptography \cite{MH78}, and game theory \cite{Hay02}.  Despite the fact that Partition can be reduced to Subset Sum, algorithms specific to Partition have been developed in the hope of solving it faster than Subset Sum. See Table~\ref{table:Partition}. Mucha, W\k{e}grzycki and W\l{}odarczyk~\cite{MWW19} gave the first subquadratic-time FPTAS for Partition that runs in $\widetilde{O}(n+(\frac{1}{\eps})^{5/3})$ time.  This result, together with the quadratic lower bound for Subset Sum, implies that Partition is easier than Subset Sum at least in terms of polynomial-time approximation schemes. On the negative side, the conditional lower bound of $\mathrm{poly}(n)/\eps^{1-o(1)}$ from Abboud, Bringmann, Hermelin, and Shabtay~\cite{ABHS22} implies that a running time of $O((n + \frac{1}{\eps})^{1-o(1)})$ is impossible assuming SETH. This naturally raises the following question.
\begin{center}
    \emph{Can Partition be approximated within $O(n + \frac{1}{\eps})$ time?}
\end{center}
Bringmann and Nakos~\cite{BN21b} gave a deterministic FPTAS that runs in $\widetilde{O}(n+(\frac{1}{\eps})^{3/2})$. The current best running time is $\widetilde{O}(n+(\frac{1}{\eps})^{5/4})$ which was obtained independently by Deng, Jin, and Mao~\cite{DJM23} and by Wu and Chen~\cite{WC22}. The question is still away from being settled.

\paragraph{Weakly Approximating Subset Sum} The study of weak approximation schemes for Subset Sum was initiated by Mucha, W\k{e}grzycki and W\l{}odarczyk~\cite{MWW19}. Weak approximation lies between strong\footnote{From now on, we say that the standard approximation is strong, in order to distinguish it from weak approximation.} approximation for Subset Sum and Partition in the sense that any weak approximation scheme for Subset Sum implies a strong approximation scheme for Partition. Consider a Subset Sum instance $(X, t)$. Given any $\eps > 0$, a weak approximation scheme finds a subset $Y$ of $X$ such that 
\[
    (1 - \eps)\Sigma(Y^*) \leq \Sigma(Y) \leq (1+\eps)t,
\]
where $Y^*$ is the optimal solution, that is, a weak approximation scheme allows a solution to slightly break the constraint (also known as resource augmentation).
Mucha, W\k{e}grzycki and W\l{}odarczyk ~\cite{MWW19} proposed an $\widetilde{O}(n+(\frac{1}{\eps})^{5/3})$-time weak approximation scheme.  It was later improved to $\widetilde{O}(n+(\frac{1}{\eps})^{3/2})$~\cite{BN21b, WC22}.  The lower bound is $(n + \frac{1}{\eps})^{1-o(1)}$, which is the same as that for Partition.

\subsection{Our Results}

\begin{theorem}\label{thm:subsetsum}
There is a deterministic $\widetilde{O}(n + \frac{1}{\eps})$-time  weak approximation scheme for Subset Sum.
\end{theorem}

Theorem~\ref{thm:subsetsum} immediately implies the following theorem since any weak approximation scheme for Subset Sum is a strong approximation scheme for Partition. 

\begin{theorem}\label{thm:partition}
     There is a deterministic $\widetilde{O}(n + \frac{1}{\eps})$-time  FPTAS for Partition.
\end{theorem}

Both of the above two results match (up to a polylogarithmic factor) the lower bound of $(n + \frac{1}{\eps})^{1-o(1)}$. To the best of our knowledge, Partition is the first NP-hard problem that admits an FPTAS that is near-linear in both $n$ and $\frac{1}{\eps}$.  We also remark that our weak approximation scheme for Subset Sum generalizes Bringmann's $\widetilde{O}(n + t)$-time exact algorithm for Subset Sum~\cite{Bri17} in the sense that when $\eps : = \frac{1}{2t}$, the weak approximation scheme becomes an exact algorithm with $\widetilde{O}(n + t)$ running time.

To attain Theorem~\ref{thm:subsetsum}, we utilize an additive combinatorics result that is different from the ones previously used in Subset Sum and Knapsack~\cite{GM91, MWW19, BN20, BW21, WC22, DJM23, CLMZ24aSODA}. We believe that it may be of independent interest.

\begin{table}[!ht]
    \centering
    \caption{Polynomial-time approximation schemes for Subset Sum. Reference with a star (*) indicates that the algorithm was designed for Knapsack and works for Subset Sum. Symbol (\dag) means that it is a randomized approximation scheme.}
    \begin{tabular}{cc}
        \toprule 
        \quad Running Time\qquad & Reference \\
        \toprule 
        \multicolumn{2}{c}{Strong approximation scheme for Subset Sum}\\
        \hline \specialrule{0em}{0pt}{2pt}
        $O(n/\eps^2)$ & *Ibarra and Kim~\cite{IK75}, *Karp~\cite{Kar75}\\
        $O(n/\eps)$ & Gens and Levner~\cite{GL78, GL79}\\
        $O(n + (\frac{1}{\eps})^4)$ & *Lawler~\cite{Law79}\\
        $O(n + (\frac{1}{\eps})^3)$ & Gens and Levner~\cite{GL94}\\
        $\widetilde{O}(n + (\frac{1}{\eps})^2)$ & Kellerer, Mansini, Pferschy and Speranza\cite{KMPS03}\\
        \hline
        C.L.B. $(n+(\frac{1}{\eps}))^{2-o(1)}$ & Bringmann and Nakos~\cite{BN21b}\\
        \toprule 
        \multicolumn{2}{c}{Weak approximation scheme for Subset Sum}\\
        \hline \specialrule{0em}{0pt}{2pt}
        $\widetilde{O}(n+(\frac{1}{\eps})^{5/3})$ &\dag Mucha, W\k{e}grzycki and W\l{}odarczyk~\cite{MWW19}\\
        $\widetilde{O}(n+(\frac{1}{\eps})^{3/2})$ & \dag Bringmann and Nakos~\cite{BN21b}, Wu and Chen~\cite{WC22}\\
        $\widetilde{O}(n + \frac{1}{\eps})$ &  This Paper \\
        \hline
        C.L.B. $\mathrm{poly}(n)/\eps^{1-o(1)}$ & Abboud, Bringmann, Hermelin and Shabtay~\cite{ABHS22}\\
    \bottomrule 
    \end{tabular}
    \label{table:SubsetSum}
\end{table}

\begin{table}[!ht]
    \centering
    \caption{Polynomial-time approximation schemes for Partition. Symbol (\dag) means that it is a randomized approximation scheme.}
    \begin{tabular}{cc}
        \toprule 
        \quad Running Time\qquad & Reference \\
        \hline \specialrule{0em}{0pt}{2pt}
        $\widetilde{O}(n + (\frac{1}{\eps})^2)$ & Gens and Levner~\cite{GL80} \\
        $\widetilde{O}(n+(\frac{1}{\eps})^{5/3})$ &\dag Mucha, W\k{e}grzycki and W\l{}odarczyk~\cite{MWW19}\\
        $\widetilde{O}(n+(\frac{1}{\eps})^{3/2})$ & Bringmann and Nakos~\cite{BN21b}\\
        $\widetilde{O}(n + (\frac{1}{\eps})^{5/4})$ & Deng, Jin and Mao~\cite{DJM23}, Wu and Chen~\cite{WC22} \\
        $\widetilde{O}(n + \frac{1}{\eps})$ & This Paper\\
        \hline
        C.L.B. $\mathrm{poly}(n)/\eps^{1-o(1)}$& Abboud, Bringmann, Hermelin and Shabtay~\cite{ABHS22}\\
        \bottomrule 
    \end{tabular}
    \label{table:Partition}
\end{table}

\subsection{Technical Overview} 
Solving the Subset Sum problem can be reduced to computing (a subset of) the collection of all possible subset sums, denoted by $\mathcal{S}(X)$. 
Let $X = \{x_1, \ldots, x_n\}$. $\mathcal{S}(X)$ can be computed by the sumset $\{x_1, 0\} + \cdots + \{x_n, 0\}$. Our main technique for efficiently approximating (a major portion of) this sumset combines a deterministic sparse convolution~\cite{BFN22} (see Lemma~\ref{lem:sparse-fft}) with an additive combinatorics result from Szemer{\'e}di and Vu~\cite{SV05} (see Theorem~\ref{thm:ap}). We briefly explain the idea below.

Suppose we want to compute the overall sumset of $A_1, \ldots, A_{\ell}$. We proceed in a tree-like manner, where the initial sets $A_1, \ldots, A_{\ell}$ form the bottom level. We compute the next level by taking the pairwise sumsets of the nodes, setting $B_i = A_{2i-1} + A_{2i}$. If $B_1, \ldots, B_{\ell/2}$ have a small total size, we compute them efficiently via sparse convolution. When their total size becomes too large to compute explicitly, we apply the additive combinatorics result to show that the sumset contains a long arithmetic progression. We further show that this progression can be extended into a long sequence with small differences between consecutive terms, which directly approximates a major part of the sumset.

To effectively use this sequence to approximate a major part of $\mathcal{S}(X)$, rounding is necessary when we compute a new level. If one level has too many subsets, then the error caused by rounding would become far more than what is acceptable. Thus, we need to ensure that we are working with a small number of subsets. That is, we need to partition $X$ into several subsets $X_1,\ldots,X_\ell$ such that a major portion of $\mathcal{S}(X)$ is in $X_1+\cdots+X_{\ell}$. A standard randomized approach for this relies on the color-coding technique, which is introduced to exact Subset Sum algorithms by Bringmann~\cite{Bri17}. To propose a deterministic algorithm, we derandomize this step by utilizing the techniques proposed by Chan~\cite{Chan26}. A direct application of Chan's algorithm is computationally too expensive for our purposes. To control the running time, we propose a novel subroutine whose running time is near-linear in the size of the output set and then combine Chan's framework with additive combinatorics.

Finally, weak approximating schemes require not only approximating the sumset values but also recovering the corresponding subsets. However, The Szemer{\'e}di-Vu theorem~\cite{SV05} is non-constructive. It guarantees the existence of a long arithmetic progression but does not provide a way to find it. This makes subset recovery difficult. An issue similarly encountered in prior dense Subset Sum algorithms (e.g., \cite{BW21}). 
While Chen, Mao, and Zhang~\cite{CMZ26} have proposed a constructive version of this theorem, their algorithm relies on randomization. To achieve a fully deterministic algorithm, we apply the Szemer{'e}di-Vu theorem purely as an analytical tool. Specifically, the theorem guarantees that once the total size is sufficiently large, the partially computed sumsets are already dense enough to approximate a major portion of the target sumset. For the remaining uncomputed subsets, we only need to retain a few elements to successfully extend the arithmetic progression.
Consequently, we can safely discard all remaining computations and proceed to the next phase. This strategy strictly bounds the running time. Furthermore, because our algorithm explicitly computes this approximating set, recovering the corresponding subsets is straightforward.

Our algorithm, and in particular its use of arithmetic progression, is inspired by the dense Subset Sum algorithm by Bringmann and Wellnitz~\cite{BW21}, and Galil and Margalit~\cite{GM91}, but we obtain the arithmetic progression in a way that is different from these two works. Bringmann and Wellnitz~\cite{BW21}, and Galil and Margalit~\cite{GM91} got the arithmetic progression via an additive combinatorics result saying that if an integer set $A$ has $\widetilde{\Omega}(\sqrt{u})$ distinct integers, where $u$ is the maximum integer in $A$, then the collection of subset sums of $A$ must contain an arithmetic progression of length ${\Omega}(u)$ {(see, e.g., \cite{Sar94})}. In contrast, we use a different result from additive combinatorics, which states that if we have $\ell$ sets $A_1, \ldots, A_\ell$ of size at least $k$  and $\ell k \geq \Omega(u)$, then $A_1 + \cdots + A_\ell$ must contain an arithmetic progression of length ${\Omega}(u)$.

\subsection{Further Related Work}
Subset Sum is a special case of Knapsack. There is a long line of research on approximation schemes for Knapsack, e.g., \cite{IK75, Kar75, Law79, KP04, Rhe15, Chan18, Jin19, DJM23}. There is a conditional lower bound of $(n + \frac{1}{\eps} )^{2-o(1)}$ based on the (min, +)-convolution hypothesis\cite{CMWW19, KPS17}. Very recently, \cite{CLMZ24bSICOMP} and \cite{Mao24} independently established an $\widetilde{O}(n+(\frac{1}{\eps})^2)$-time FPTAS for Knapsack.

In addition to approximation algorithms, exact pseudopolynomial-time algorithms for Knapsack and Subset Sum have also received extensive studies in recent years, e.g., \cite{Bri17, AT19, PRW21, BC23, CLMZ24aSODA, Bri24, Jin24}. 
It is worth mentioning that there is an algorithm~\cite{BN20} with running time $\widetilde{O}({k^{4/3}})$ for Subset Sum where $k$ is the number of subset sums that are smaller than $t$, which shares a similar flavor to our algorithm in the sense that it combines sparse convolution and Ruzsa's triangle inequality for sumset estimation (which is a different result in additive combinatorics).

Our algorithm utilizes sparse convolution algorithms, see, e.g.,~\cite{CH02, AR15, CL15, Nak20, BFN21, BN21a, BFN22}. We also adopt results from additive combinatorics, see, e.g., \cite{Alo87, Sar89, Sar94, SV05, SV06,CMZ25,CMZ26}. Additive combinatorics results have been exploited extensively in recent years on Knapsack and Subset Sum, see, e.g., \cite{GM91, MWW19, BN20, BW21, WC22, DJM23, CLMZ24aSODA,CLMZ24bSICOMP}.

\subsection{Paper Organization} 

Section~\ref{sec:pre} introduces the necessary terminology and analytical tools, demonstrating that it suffices to solve a reduced version of the problem. In Section~\ref{sec:subsetsum-fptas-dense}, we apply tools from additive combinatorics to establish structural properties for the dense case. Section~\ref{sec:subsetsum-fptas-det-partition} details how we carefully partition the set by utilizing techniques proposed by Chan with the additive combinatorics tools. Section~\ref{sec:base-layer} focuses on computing the base layer of our tree structure, proving that the sum of these sets approximates the major part of the target sumset $\mathcal{S}(X)$. We then complete the sumset computation by merging these subsets up the tree in Section~\ref{sec:subsetsum-fptas-tree-merge}. Finally, Section~\ref{sec:conclude} concludes the paper.

\section{Preliminaries}\label{sec:pre}
\subsection{Notation and Definitions}
Throughout this paper, we assume that $\eps > 0$ is  sufficiently small. We also assume that $\frac{1}{\eps}$ is an integer by adjusting $\eps$ by an $O(1)$ factor. 
All logarithms ($\log$) in this paper are base $2$.

Let $w,v$ be two real numbers.  We use $[w,v]$ to denote the set of integers between $w$ and $v$.  That is, $[w,v] = \{z \in \mathbb{Z} : w\leq z\leq v\}$.  Let $Y$ be a nonempty set of integers. We write $Y \cap [w,v]$ as $Y[w,v]$. We denote the minimum and maximum elements of $Y$ by $\min(Y)$ and $\max(Y)$, respectively. The diameter of $Y$ is defined as $\dm(Y)=
\max(Y)-\min(Y)+1$. For emptysets, we define $\max(\emptyset)=\min(\emptyset) = 0$ and $\dm(\emptyset)=0$.  We write $\sum_{y \in Y}y$ as $\Sigma(Y)$. We refer to the number of elements in $Y$ as the size of $Y$, denoted by $|Y|$. Throughout this paper, we use both the terms ``set'' and ``multi-set''. Only when the term ``multi-set'' is explicitly used do we allow duplicate elements. Unless otherwise stated, a subset of a multi-set is a multi-set.

Let $(X, t)$ be an arbitrary Subset Sum instance. We assume that $x \leq t$ for any $x \in X$ since the integers greater than $t$ can be safely removed from $X$.    We define $\mathcal{S}(X)$ to be the set of all subset sums of $X$. That is, $\mathcal{S}(X) = \{\Sigma(Y) : Y\subseteq X\}$. Subset Sum actually asks for the maximum element of $\mathcal{S}(X)[0,t]$. 

We consider the more general {problem} of weakly approximating $\mathcal{S}(X)$.

\begin{definition}\label{def:approx}
    Let $S$ be a set of integers, let $w,v$ be two real numbers and let $\delta\geq 0$. We say a set $\widetilde{S}$ approximates $S[w,v]$ with additive error $\delta$ if 
    \begin{enumerate}[label={\normalfont(\roman*)}]
        \item for any $s \in S[w,v]$, there is $\tilde{s} \in \widetilde{S}$ with $s - \delta\leq \tilde{s} \leq s + \delta$, and

        \item for any $\tilde{s} \in \widetilde{S}$, there is $s \in S$ with $\tilde{s} -\delta \leq s \leq \tilde{s} + \delta$.
    \end{enumerate}
\end{definition}
Condition~(i) ensures that the approximate set contains a value that approximates the exact optimal solution, while condition~(ii) ensures that every value in the approximate set has a true counterpart in $S$. Notice that in condition~(ii), the witnessing element $s\in S$ is not required to lie in $[w,v]$. This is intentional: approximate values near the boundary of $[w,v]$ may correspond to exact value just outside the interval.

\begin{lemma}\label{lem:weak-approx-scheme}
    Let $(X, t)$ be a Subset Sum instance. Suppose that for any given $\eps > 0$, an algorithm can compute a set $\widetilde{S}$ that approximates $\mathcal{S}(X)[0, t]$ with an additive error of $\eps t$, and for any $\tilde{s} \in \widetilde{S}$, it can recover a subset $Y \subseteq X$ such that $\tilde{s} - \eps t \leq \Sigma(Y) \leq \tilde{s} + \eps t$. Then, this algorithm is a weak approximation scheme for Subset Sum.
\end{lemma}
\begin{proof}
    Let $Y^*$ be an optimal subset such that $\Sigma(Y^*) = \max (\mathcal{S}(X)[0, t])$. Let $\tilde{s}$ be the maximum element in $\widetilde{S}[0, (1 + \eps) t]$. By definition, $\tilde{s} \leq (1 + \eps) t$. By Definition~\ref{def:approx}(i), there must be some element in $\widetilde{S}$ that is at least $\Sigma(Y^*) - \eps t$. Since $\Sigma(Y^*) \leq t$, this approximate element is within $[0, (1+\eps)t]$. This implies that the maximum $\tilde{s}$ satisfies $\tilde{s} \geq \Sigma(Y^*) - \eps t$. By Definition~\ref{def:approx}(ii), there exists a subset $Y \subseteq X$ such that $\tilde{s} - \eps t \leq \Sigma(Y) \leq \tilde{s} + \eps t$. The algorithm can recover such a $Y$ from $\tilde{s}$. 
    Therefore, we have 
    $$(1 - 4\eps)\Sigma(Y^*) \leq \Sigma(Y^*) - 2\eps t \leq \tilde{s} - \eps t \leq \Sigma(Y) \leq \tilde{s} + \eps t \leq t + 2\eps t.$$
    The first inequality holds because we can safely assume the optimal objective value satisfies $\Sigma(Y^*) \geq t/2$.\footnote{If $\Sigma(Y^*) < t/2$, every integer in $X$ must be less than $t/2$. This implies that $Y^* = X$; otherwise, we could improve $Y^*$ by adding another integer from $X$. Such an instance can be solved trivially in $O(n)$ time.} By adjusting $\eps$ by a factor of $4$, the bounds become $$(1 - \eps)\Sigma(Y^*) \leq \Sigma(Y) \leq (1 + \eps) t.$$
    This exactly matches the definition of a weak approximation scheme.
\end{proof}

\subsection{Sumset}
Let $X_1,X_2$ be a partition of $X$.  We have $\mathcal{S}(X) = \mathcal{S}({X_1}) + \mathcal{S}({X_2})$, where the sum of two sets is defined by the following.
\[
    A + B = \{a + b : a \in A, b\in B\}.
\]
The set $A + B$ is called the sumset of $A$ and $B$.  It is well-known that the sumset of two integer sets can be computed via the classical convolution algorithm based on Fast Fourier Transformation (FFT). 
\begin{lemma}\label{lem:fft}
    Let $u$ be a positive integer. Let $A$ and $B$ be two subsets of $[0,u]$. We can compute their sumset $A + B$ in $O(u\log u)$ time.
\end{lemma}

When the sumset has a small size, sparse convolution algorithms whose running time is linear in the output size may be used. We use the following deterministic algorithm due to Bringmann, Fischer, and Nakos~\cite{BFN22}.
\begin{lemma}[\normalfont \cite{BFN22}]\label{lem:sparse-fft}
    Let $u$ be a positive integer. Let $A$ and $B$ be two subsets of $[0,u]$. We can compute their sumset $A + B$  in $O(|A + B| \log^5 u\,\mathrm{polyloglog}\,u)$ time.
\end{lemma}

Using the method proposed by Alon and Naor~\cite{AN96}, we obtain a strengthened version: in addition to computing the sumset, we can perform backtracking on the elements of the sumset.

\begin{lemma}[\cite{Chan26}]\label{lem:subsetsum-fptas-sumset-compute-u}
    Let $u$ be a positive integer. Let $A$ and $B$ be two subsets of $[0,u]$. We can compute their sumset $A + B$  in $O(u\log u)$ time. Moreover, for any $s\in A+B$, we can find $a\in A$ and $b\in B$ with $a+b=s$ in $O(1)$ time.
\end{lemma}

\begin{lemma}[\cite{Chan26}]\label{lem:subsetsum-fptas-sumset-compute}
    Let $u$ be a positive integer. Let $A$ and $B$ be two subsets of $[0,u]$. We can compute their sumset $A + B$  in $O(|A + B| \log^5 u\,\mathrm{polyloglog}\,u)$ time. Moreover, for any $s\in A+B$, we can find $a\in A$ and $b\in B$ with $a+b=s$ in $O(1)$ time.
\end{lemma}

\begin{remark}
    Throughout this paper, all sumset computations rely on Lemma~\ref{lem:subsetsum-fptas-sumset-compute-u} and Lemma~\ref{lem:subsetsum-fptas-sumset-compute}. Consequently, once we obtain a set $\widetilde{S}$ that approximates $\mathcal{S}(X)$, for any $\tilde{s} \in \widetilde{S}$, we can efficiently backtrack to recover a subset $Y \subseteq X$ satisfying $\tilde{s} - \eps t \leq \Sigma(Y) \leq \tilde{s} + \eps t$ in $\widetilde{O}(n + \frac{1}{\eps})$ time. Thus, it suffices for the remainder of the paper to focus on how to compute the approximating set $\widetilde{S}$ with an additive error of $\eps t$.
\end{remark}

Lemma~\ref{lem:subsetsum-fptas-sumset-compute-u} immediately implies an approach for approximating the sumset of several sets. The proof of this lemma closely follows the technique introduced by Chan~\cite{Chan18}.

\begin{lemma}\label{lem:approx-k-fft}
    Let $u$ be a positive integer. Let $A_1,\ldots, A_{\ell}$ be subsets of $[0,u]$ with total size $k$. For any $\eps < 1$, in $O(k + \frac{\ell^2}{\eps}\log \frac{\ell}{\eps})$ time, we can compute a set $S$ of size $O(\frac{1}{\eps})$ that approximates $(A_1+\cdots+A_{\ell})[0, u]$ with additive error $\eps u$.
\end{lemma}
\begin{proof}
    We first round the elements of all $A_i$'s down to the nearest multiple of $\eps u$, and delete duplicate elements in each $A_i$. This step takes $O(k)$ time and incurs an additive error of at most $\eps u$ for each $A_i$. Then we can compute $(A_1+\cdots+A_{\ell})[0, u]$ via Lemma~\ref{lem:fft}, and the total running time is $O(\frac{\ell}{\eps}\log \frac{1}{\eps})$. The total additive error is $\ell \eps u$.  
    
    By adjusting $\eps$ by a factor of $\ell$, the total additive error becomes $\eps u$ and the running becomes $O(k + \frac{\ell^2}{\eps}\log \frac{\ell}{\eps})$. Rounding the elements of the output set down to the nearest multiple of $\eps u$ reduces its size to $O(\frac{1}{\eps})$ and incurs an additive error of at most $\eps u$.
\end{proof}

\subsection{Problem Reduction}\label{sec:subsetsum-fptas-reduce-problem}
We first show that, by combining the approximate convolution method from the previous section with some standard scaling and rounding techniques, we can reduce the Subset Sum problem to the following reduced problem $\mathrm{RP}(\kappa)$.

\begin{definition}[The Reduced Problem $\mathrm{RP}(\kappa)$]
    Let $X$ be a multi-set of integers from $[\frac{1}{\eps}, \frac{2}{\eps})$ such that $\Sigma(X) \geq \frac{4\kappa}{\eps}$. Compute a set $\widetilde{S}$ of size $\widetilde{O}(\frac{1}{\eps})$ that approximates $\mathcal{S}(X)[\frac{\kappa}{\eps}, \frac{2\kappa}{\eps}]$ with additive error $\delta= {O}(\kappa\,\, \mathrm{polylog}(\frac{n}{\eps}))$.
\end{definition}
\begin{restatable}{lemma}{lemreduce}
\label{lem:subsetsum-fptas-reduce}
    There is an $\widetilde{O}(n + \frac{1}{\eps})$-time weak approximation scheme for Subset Sum if, for any $\kappa\in[2,\frac{1}{\eps}]$, the reduced problem $\mathrm{RP(\kappa)}$ can be solved in $\widetilde{O}(n + \frac{1}{\eps})$ time.
\end{restatable}


\begin{proof}

Let $(X, t)$ be the original instance, and define $Z = \{x \in X : x < \eps t\}$ as the set of small items. If $\Sigma(Z) \geq \eps t$, we iteratively extract disjoint minimal subsets $Z_1, Z_2, \ldots, Z_h$ from $Z$ such that $\Sigma(Z_j) \geq \eps t$ for each $j$. This process continues until the remaining elements form a residual set $R \subseteq Z$ with $\Sigma(R) < \eps t$. We can simply discard $R$, which incurs an additive error of at most $\eps t$. Next, we replace each extracted subset $Z_j$ with a single integer of value $\Sigma(Z_j)$. Note that by the minimality of $Z_j$, we have $\Sigma(Z_j) < 2\eps t$ for all $j$. Therefore, the additive error caused by utilizing these replacement integers is at most $2\eps t$, since for any subset $S \subseteq Z$, there exists a subset of the replacement integers whose sum is at least $\Sigma(S) - 2\eps t$. The entire transformation can be implemented in $O(n)$ time. 

Now for all $x \in X$ we have $x \geq \eps t$, i.e., $\frac{x}{\eps^2t} \geq \frac{1}{\eps}$. By adjusting $\eps$ by a constant factor, we assume $\frac{1}{\eps}$ is an integer. Now we scale the entire instance by $\eps^2 t$, i.e., we replace each $x\in X$ with $x':= \lfloor \frac{x}{\eps^2 t}\rfloor$ and replace $t$ with $t' := \lfloor \frac{t}{\eps^2t}\rfloor$. The relative error caused by scaling is at most $1 + \eps$, so the additive error is at most $\eps t$. After scaling, for all $x \in X$ we have $x \in [\frac{1}{\eps}, (\frac{1}{\eps})^2)$, and $t = (\frac{1}{\eps})^2$.

We then partition $X$ into $\log \frac{1}{\eps}$ groups, where each group contains integers in $[\frac{\alpha}{\eps}, \frac{2\alpha}{\eps})$ for $\alpha \in \{1,2,4,8,\cdots\} \cap [1, \frac{1}{\eps}]$. We write $X \cap [\frac{\alpha}{\eps}, \frac{2\alpha}{\eps})$ as $X_{\alpha}$. Suppose for each $X_{\alpha}$, we can compute a set $\tilde{S}_{\alpha}$ of size $\widetilde{O}(\frac{1}{\eps})$ that approximates $\mathcal{S}(X_{\alpha})[0,t]$ in $\widetilde{O}(n+\frac{1}{\eps})$ time. Then we can compute a set $\widetilde{S}$ of size $\widetilde{O}(\frac{1}{\eps})$ that approximates $\mathcal{S}(X)[0,t]$ with additive error $\widetilde{O}(\eps t)$ in $\widetilde{O}(n+\frac{1}{\eps})$ time via Lemma~\ref{lem:approx-k-fft} (taking $A_j$ to be $\tilde{S}_{\alpha}$). 

Now we consider $X_{\alpha}$, where $x\in X_{\alpha}$ satisfies $x \in [\frac{\alpha}{\eps}, \frac{2\alpha}{\eps})$ and $\alpha \in [1, \frac{1}{\eps}]$. For each $x\in X_{\alpha}$, we further scale by $\alpha$. That is, for each $x \in X_{\alpha}$ we let $x'=\lfloor x/\alpha\rfloor$ and let $t' = \lfloor\frac{1}{\alpha\eps^2}\rfloor$. Again, the relative error caused by scaling is at most $1 + \eps$, so the additive error is at most $\eps t$.

Now for each $x\in X_{\alpha}$ we have $x\in [\frac{1}{\eps}, \frac{2}{\eps})$, and $t=\lfloor\frac{1}{\alpha\eps^2}\rfloor\in [\frac{1}{\eps},(\frac{1}{\eps})^2]$. 
The following claim shows that it actually suffices to approximate $\mathcal{S}(X_{\alpha})[0,\Sigma(X_{\alpha})/2]$.

\begin{claim}\label{claim:half}
Let $\widetilde{S}$ be a set that approximates $\mathcal{S}(X)[0,\Sigma(X)/2]$ with additive error $\eps t$. Then $\{\Sigma(X)-s':s'\in \widetilde{S}\}$ approximates $\mathcal{S}(X)[\frac{\Sigma(X)}{2}, \Sigma(X)]$ with additive error $\eps t$.
\end{claim}

\begin{proof}
    We prove Claim~\ref{claim:half} according to Definition~\ref{def:approx}.
Consider any $s\in \mathcal{S}(X)[\Sigma(X)/2, \Sigma(X)]$. We have $\Sigma(X)-s\in \mathcal{S}(X)[0, \Sigma(X)/2]$. Therefore, there exists $\tilde{s}\in \widetilde{S}$ with $\Sigma(X)-s-\eps t\leq \tilde{s}\leq \Sigma(X)-s+\eps t$. Rearranging the inequality, we get $s-\eps t \leq \Sigma(X)-\tilde{s}\leq s+\eps t$. 
Conversely, let $\Sigma(X)-\tilde{s}$ be an element of the reflected set, where $\tilde{s} \in \widetilde{S}$. By Definition~\ref{def:approx}(ii), there exists $s \in \mathcal{S}(X)$ such that $\tilde{s} -\eps t \leq s \leq \tilde{s} + \eps t$. Then $\Sigma(X)-s\in \mathcal{S}(X)$, and $\Sigma(X)-\tilde{s} -\eps t \leq \Sigma(X)-s \leq \Sigma(X)-\tilde{s} + \eps t$. Thus condition (ii) also holds. Therefore, $\{\Sigma(X)-s':s'\in \widetilde{S}\}$ approximates $\mathcal{S}(X)[\Sigma(X)/2, \Sigma(X)]$ with additive error $\eps t$.
\end{proof}

Therefore, we can assume $t\leq \Sigma(X_{\alpha})/2$. To approximate $\mathcal{S}(X_{\alpha})[0,t]$, we partition $\mathcal{S}(X_{\alpha})[0,t]$ into $\log t \leq 2\log \frac{1}{\eps}$ subsets $\mathcal{S}(X_{\alpha})[0,\frac{1}{\eps}], \mathcal{S}(X_{\alpha})[\frac{1}{\eps},\frac{2}{\eps}],\mathcal{S}(X_{\alpha})[\frac{2}{\eps},\frac{4}{\eps}],\ldots$. Note that $\mathcal{S}(X_{\alpha})[0,\frac{1}{\eps}]$ and $\mathcal{S}(X_{\alpha})[\frac{1}{\eps},\frac{2}{\eps}]$ can be computed directly, since for each $x\in X_{\alpha}$ we have $x\in [\frac{1}{\eps}, \frac{2}{\eps})$. Each of the remaining subsets can be written as $\mathcal{S}(X_{\alpha})[\frac{\kappa}{\eps}, \frac{2\kappa}{\eps}]$ for some integer $\kappa\in [2, \frac{1}{\eps}]$. To approximate $\mathcal{S}(X_{\alpha})[0,t]$ with additive error $\widetilde{O}(\eps t)$, it suffices to approximate each $\mathcal{S}(X_{\alpha})[\frac{\kappa}{\eps}, \frac{2\kappa}{\eps}]$ with additive error $\widetilde{O}(\kappa)$. Note that since $\frac{2\kappa}{\eps}\leq t$ and $t\leq \Sigma(X_{\alpha})/2$, so $\Sigma(X_{\alpha})\geq \frac{4\kappa}{\eps}$ for each $\mathcal{S}(X_{\alpha})[\frac{\kappa}{\eps}, \frac{2\kappa}{\eps}]$.
\end{proof}

\section{The Dense Case}\label{sec:subsetsum-fptas-dense}

In this section, we focus on the ``dense'' case, which occurs when the total size of the sets is large enough to exhibit useful structural properties.
We first introduce an additive combinatorial result showing that a sufficiently large total size forces the sumset of a small subcollection of the sets to contain a long arithmetic progression.
Then we show how to extend this progression into a long sequence with  small differences between consecutive terms, and show that this sequence provides a good approximation for the overall sumset.

\subsection{Additive combinatorics tools}

The following lemma is the main tool we will use in the dense case. It shows that given many sets with a large total size, their sumset must have a long arithmetic progression. If we require the total size to be roughly $\rho$ times the threshold, we can select a group of sets such that the sumset of any $1/\rho$ fraction of them has an arithmetic progression. In practice, we always choose the sets with small weights (small $f$ values).

\begin{lemma}[{\cite[Lemma~14]{CLMZ26MP}}]
\label{lem:ap-with-small-element}
    There is a constant $C \geq 1$ such that the following holds. Let $A_1, \ldots, A_{\ell}$ be $\ell$ sets of integers with maximum diameter $u$ and let $f: [1,\ell] \to \mathbb{N}$ be a weight function. If $\sum_{i=1}^\ell |A_i| \geq \ell + C\rho u' \log u'$ for some $u' \geq u$ and some $\rho \geq 1$, then we can select an index set $I$ such that $|I|\leq \frac{\ell}{\rho}$ and the sumset $\sum_{i \in I} A_i$ contains an arithmetic progression of length at least $u'$, while satisfying 
    \[
        \sum_{i\in I}f(i) \leq \frac{1}{\rho}\sum_{i=1}^{\ell}f(i).
    \]
\end{lemma}

Lemma~\ref{lem:ap-with-small-element} is proposed by \cite{CLMZ26MP}; we reproduce its proof here for completeness. This Lemma is derived from the following theorem of Szemer{\'e}di and Vu.

\begin{theorem}[implied by{\cite[Corollary~5.2]{SV05}}]\label{thm:ap}
    There exist positive constants $c_1, c_2$ such that the following holds. Let $A_1, \ldots, A_{\ell}$ be subsets of $[1,u]$, each of size $k$. If $\ell k \geq c_1 u$, then $A_1 + \cdots + A_\ell$ contains an arithmetic progression of length at least $c_2 \ell k$.
\end{theorem}

We use three refinements to make Theorem~\ref{thm:ap} more flexible for our purposes. First, instead of imposing a lower bound on the size of every individual set, we work with the total size of all the sets. Second, we introduce a parameter $\rho$: when the total size exceeds the required threshold by a factor of $\rho$, it suffices to use only a $1/\rho$ fraction of the sets, and this subcollection can be chosen so as to minimize an arbitrary weight function $f$. Finally, by translating the sets appropriately, 
we can make the density requirement depend on the maximum diameter of the sets rather than on their maximum elements.

\begin{proof}
Let $C$ be sufficiently large such that $C \geq 4c_1$ and that $C c_2\geq 4$ where $c_1,c_2$ are the constants from Theorem~\ref{thm:ap}.
  The case $u = 1$ is trivial. Assume $u \ge 2$. For each $i$, let $a_i := \min A_i$ and define $A'_i := A_i - a_i + 1$, a non-empty subset of $[1, u]$ (since $u$ is the maximum diameter of the $A_i$'s).

Since $u'\geq u$, $A'_1,\ldots,A'_{\ell} \subseteq[0,u']$.
  We claim that there exists some $k\in [2,u]$ such that at least $\frac{C\rho u'  }{2k}$ sets from $\{A'_1, \ldots, A'_\ell\}$ have size at least $k$. Assume the claim holds. We have $\ell \geq \frac{C\rho u'  }{2k}$. We greedily select $\lceil\frac{Cu'  }{4k}\rceil$ of these $A'_i$'s by prioritizing those that minimize $f(i)$. Let $I$ be the index set of the selected $A'_i$'s. We have $|I|\leq \frac{\ell}{\rho}$ and
  $|I|k\geq Cu'/4\geq   c_1 u'$,
  so $\{A'_i\}_{i\in I}$ satisfies the condition of Theorem~\ref{thm:ap}. Then $\{A'_i\}_{i\in I}$ contains an arithmetic progression of length at least $c_2|I|k\geq u'$. Since translating each $A'_i$ back to $A_i$ only shifts the sumset, $\{A_i\}_{i\in I}$ also contains an arithmetic progression of length at least $u'$.

    Moreover, since we greedily select $\lceil\frac{Cu'  }{4k}\rceil$ of these $A_i$'s by prioritizing those that minimize $f(i)$, we have
    \[
    \frac{\sum_{i\in I}f(i)}{\sum_{i=1}^{\ell}f(i)}\leq \frac{\lceil Cu' /4k\rceil}{C\rho u'  /2k}\leq \frac{1}{\rho}.
    \]

    We prove the claim by contradiction.  Let $\ell_k$ be the number of sets from $\{A'_1, \ldots, A'_\ell\}$ that have size at least $k$. Suppose that $\ell_k < \frac{C\rho u'  }{2k}$ for any $k \in [2, u]$.
    Then 
    \begin{align*}
        \sum_{i=1}^{\ell}|A'_i| &= \sum_{k=1}^{u} k(\ell_k - \ell_{k+1}) \\&=\sum_{k=1}^{u}\ell_k \\&<\ell_1 + \sum_{k=2}^{u} \frac{C\rho u'  }{2k} \\&\leq \ell + C\rho u'   (1 + \log u)/2 \\&\leq \ell + C\rho u'  \log u'.
    \end{align*}
    Contradiction. (The last inequality is due to our assumption that $u \geq 2$.) 
\end{proof}

\subsection{Density in Approximation Algorithms}

If there exists a sequence whose consecutive terms differ by a sufficiently small amount and that covers the target interval, then every value in that interval can be approximated.

\begin{lemma}\label{lem:subsetsum-fptas-sequence-appro-set}
Let $\delta$ be a positive integer, $k>1$ be an integer, and $S$ be a set of integers, and let $S'=\{a_1,\ldots,a_k\}\subseteq S$ with $0< a_i - a_{i-1} \leq \delta$ for any $i \in [2, k]$. Then for any real numbers $w$ and $v$ with $a_1 \leq w$ and $v \leq a_k$, $S'$ approximates $S[w, v]$ with additive error $\delta$.
\end{lemma}

\begin{proof}
    For any $s\in S[w,v]$, there must exist some $j\in[1,k-1]$ with
    \[
      a_j\leq s\leq a_{j+1}\leq a_{j}+\delta,
    \]
    hence $s-\delta\leq a_j\leq s$.
    The first condition of Definition~\ref{def:approx} is satisfied. The second holds since $S'\subseteq S$. Therefore, $S'$ approximates $S[w,v]$ with additive error $\delta$.
\end{proof}


In view of Lemma~\ref{lem:subsetsum-fptas-sequence-appro-set}, a sequence with sufficiently small differences between consecutive terms of subset sums is enough to approximate a set. The next lemma shows that it can be obtained by extending an arithmetic progression. 

\begin{lemma}\label{lem:ap-extend}
    Let $B$ be a set of positive integers that contains an arithmetic progression $b_1 < \ldots < b_k$ with common difference $\Delta$. Let $u$ be a positive integer. Given subsets $A_1, \ldots, A_{\ell}\subseteq [0,u]$ where $0\in A_i$ for all $i\in[1,\ell]$, let $\eta = \sum_{i=1}^\ell \max(A_i)$ and $S = B + A_1 + \cdots + A_\ell$.
    If $b_k - b_1 \geq u - 1$, then $S$ contains a sequence $b_1 = s_1 < \cdots < s_{k'} = b_k + \eta$ such that $s_i - s_{i-1} \leq \Delta$ for $i \in [2, k']$.
\end{lemma}
\begin{proof}
    We prove by induction on $j$.  When $j = 0$, the lemma is obviously true by letting $k' = k$ and $s_i = b_i$ for $i\in [1,k]$. Suppose that the lemma is true for some $0\leq j  \leq \ell -1$. We prove that it is also true for $j+1$.  Let $S_j = B + A_1 + \cdots + A_j$ and $\eta_j = \max(A_1) + \cdots + \max(A_j)$. Let $S_{j+1}= S_j + A_{j+1}$ and $\eta_{j+1} = \eta_j + \max(A_{j+1})$.  By inductive hypothesis, $S_{j}$ contains a sequence $(s_1, \cdots, s_{k'})$ with $s_1 = b_1$, $s_{k'} = b_k + \eta_j$, and $s_i - s_{i-1} \leq \Delta$ for $i \in [2, k']$.  Note that $S_{j} \subseteq S_{j+1}$ since $0 \in A_{j+1}$, so $S_{j+1}$ also contains the sequence $(s_1, \ldots, s_{k'})$.  Let $a^*$ be the maximum element in $A_{j+1}$.  The sequence $(s_1 + a^*, \ldots, s_{k'} + a^*)$ also belongs to $S_{j + 1}$. Note that $s_1 + a^* \leq b_1 + u \leq b_k + 1 \leq s_{k'} + 1$ and that $s_{k'} + a^* = b_k + \eta_j + a^* = b_k + \eta_{j+1}$.  Therefore, merging the two sequences by taking union and delete duplicates yields a sequence $z_1, \ldots, z_{k''}$ in $S_{j+1}$ with $z_1 = b_1$, $z_{k''} = b_k + \eta_{j+1}$, and $z_i - z_{i-1} \leq \Delta$ for any $i \in [2,k'']$.
\end{proof}

In summary, these three lemmas collectively resolve the dense case by yielding the desired approximation of $\mathcal{S}(X)[\frac{\kappa}{\eps},\frac{2\kappa}{\eps}]$. Lemma~\ref{lem:ap-with-small-element} identifies an arithmetic progression within a subset $X'\subseteq X$ where the minimum element is at most $\frac{\kappa}{\eps}$ and the common difference is at most $\kappa$. By extending this progression using the elements in $X \setminus X'$ (Lemma~\ref{lem:ap-extend}), we obtain a sequence that completely covers the target interval $[\frac{\kappa}{\eps},\frac{2\kappa}{\eps}]$ with consecutive differences of at most $\kappa$. As shown by Lemma~\ref{lem:subsetsum-fptas-sequence-appro-set}, this sequence approximates $\mathcal{S}(X)[\frac{\kappa}{\eps},\frac{2\kappa}{\eps}]$ with an additive error of $O(\kappa)$, thereby concluding the dense case analysis.

\section{Partitioning the Set}\label{sec:subsetsum-fptas-det-partition}

Recall from Section~\ref{sec:subsetsum-fptas-dense} that the arithmetic progression requires requires controlling its minimum element (at most $\frac{\kappa}{\eps}$) and its common difference (at most $\kappa$). Since the common difference is bounded by the maximum element divided by its length (which is $O(1/\eps)$), it suffices to restrict the maximum element to $O(\kappa/\eps)$. This maximum element is bounded by the sum of the maximum elements across the sets forming the arithmetic progression. In the approximation scheme, rescaling and rounding allow us to control the maximum values. Thus, it remains to ensure that the number of sets is on the order of $\kappa$. 
Here we present a deterministic algorithm to ensure it. The framework was proposed by~\cite{Chan26}.

Let $X\subseteq[\frac{1}{\eps},\frac{2}{\eps})$, and the goal is to approximate $\mathcal{S}(X)[\frac{\kappa}{\eps},\frac{2\kappa}{\eps}]$. For any $Y\subseteq X$ with $\Sigma(Y)\in \mathcal{S}(X)[\frac{\kappa}{\eps},\frac{2\kappa}{\eps}]$, clearly $|Y\cap X|=|Y|\leq 2\kappa$; we say the set $X$ contributes at most $2\kappa$ elements to $Y$, or simply that the set $X$ contributes at most $2\kappa$ elements. The technique in~\cite{Chan26} can deterministically ``evenly'' partition the set. Specifically, one can partition $X$ into $X_1$ and $X_2$ such that for any $s\in \mathcal{S}(X)[\frac{\kappa}{\eps},\frac{2\kappa}{\eps}]$, there exists $Y\subseteq X$ with $\Sigma(Y)=s$ and $|Y\cap X_1|\leq \kappa+2\kappa/\log 2\kappa$, $|Y\cap X_2|\leq \kappa+2\kappa/\log 2\kappa$. After that, we only need to consider the case where $X_1$ and $X_2$ each contributes $\kappa+2\kappa/\log 2\kappa$ elements. By iteratively partitioning $\lceil\log \kappa\rceil$ times, one can deterministically partition the set into $\widetilde{O}(\kappa)$ subsets, each contributing $\widetilde{O}(1)$ elements. 

Directly applying Chan's algorithm incurs an excessive computational cost. To tightly control the running time, we propose a subroutine that runs in time near-linear with respect to the output size, and bridges Chan's framework with our density analysis. If the total size of the output sets is sufficiently large, we show that a subset $X' \subseteq X$ already contains the desired arithmetic progression, meaning we only need to partition this subset $X'$. (In the following section, we use the remaining elements in $X \setminus X'$ to extend this progression.) Otherwise, we efficiently partition the entire set $X$ within the desired time bound.

For a set $X$, we use $\mathcal{S}_k(X)$ to denote the set of sums of all $k$-element subsets of $X$, i.e., $\mathcal{S}_k(X)=\{\Sigma(Y):Y\subseteq X, |Y|=k\}$, and let $\mathcal{S}_0(X)=\{0\}$. Define $\mathcal{S}_{\leq k}(X)=\bigcup_{i=0}^{\lfloor k\rfloor} \mathcal{S}_i(X)$. Since $X\subseteq[\frac{1}{\eps},\frac{2}{\eps})$, computing $\mathcal{S}(X)[\frac{\kappa}{\eps}, \frac{2\kappa}{\eps}]$ only requires computing $\mathcal{S}_{\leq 2\kappa}(X)$. Let $c_1,c_2$ be sufficiently large constants.

The main result of this section is the following:

\begin{restatable}{lemma}{lemsubsetsumfptasdetpartition}
    \label{lem:subsetsum-fptas-det-partition}
    Given a multi-set $X$ of $n$ integers from $[\frac{1}{\eps}, \frac{2}{\eps})$ and a positive integer $\kappa\in[2,\frac{1}{\eps}]$, with $\Sigma(X)\geq \frac{4\kappa}{\eps}$, in $\widetilde{O}(n+\kappa+\frac{1}{\eps})$ time, we can either
    \begin{enumerate}[label={\normalfont(\roman*)}]
        \item find a subset $X'\subseteq X$ such that $\Sigma(X')\leq \frac{1}{4}\Sigma(X)$ and $\mathcal{S}(X')$ contains an arithmetic progression $a_1 < \cdots < a_{u'}$ with $u' \geq \frac{2}{\eps}$ and $a_1 < a_{u'} \leq \frac{\kappa}{\eps}$, and partition $X'$ into $X'_1,\ldots,X'_{4\kappa}$, satisfying
        \[
            \mathcal{S}_{\leq 2\kappa}(X') \subseteq \sum_{i=1}^{4\kappa} \mathcal{S}_{\leq c_1\log^{c_2}\kappa}(X'_i);
        \]or
        \item partition $X$ into $X_1,\ldots,X_{4\kappa}$, satisfying $\sum_{i=1}^{4\kappa} \max(X_i) \geq \frac{4\kappa}{\eps}$ and
        \[
            \mathcal{S}_{\leq 2\kappa}(X) \subseteq \sum_{i=1}^{4\kappa} \mathcal{S}_{\leq c_1\log^{c_2}\kappa}(X_i).
        \]
    \end{enumerate}
\end{restatable}

\subsection{Almost-Even Partitioning}

To evenly partition $X$, we use the following lemma from discrepancy theory. Although the original lemma applied to set, we can use it for a multi-set of integers by uniquely indexing identical elements.

\begin{lemma}[{\cite[Lemma~2.3]{Chan26}}]\label{lem:subsetsum-fptas-partition-basic-lemma}
Given a multi-set $X$ and a collection $\mathcal{C}$ of subsets of $X$ each of size at most $k$, we can compute a
partition of $X$ into two subsets $X'$ and $X''$ such that $|S \cap X'|, |S \cap X''| \leq |S|/2+O(\sqrt{|S|\log |\mathcal{C}|})$ for all $S \in \mathcal{C}$, in $O(k|\mathcal{C}|+|X|)$ deterministic time.
\end{lemma}

This lemma states that given a collection of subsets of $X$, we can partition $X$ into two subsets such that every subset in the collection is partitioned almost evenly. The possible deviation depends on the sizes of the subsets in the collection, and the running time depends on the number of subsets in the collection.
To make full use of this lemma, Chan constructed a collection of small ``canonical'' subsets and showed that any element $s\in \mathcal{S}(X)[\frac{\kappa}{\eps},\frac{2\kappa}{\eps}]$ corresponds to a subset $Y\subseteq X$ with $\Sigma(Y)=s$ that can be decomposed into several small ``canonical'' subsets. Using the above lemma, all small ``canonical'' subsets are ``evenly'' partitioned, so that for every element $s\in \mathcal{S}(X)[\frac{\kappa}{\eps},\frac{2\kappa}{\eps}]$, there exists a subset $Y\subseteq X$ with $\Sigma(Y)=s$ that is ``evenly'' partitioned.

We first formalize the definition of the small ``canonical'' subsets mentioned above. All subsets are understood as subsets of the indexed ground set of the multi-set $X$; in particular, ``disjoint'' means disjoint as indexed subsets.

\begin{definition}[implied by {\cite{Chan26}}]
    Let $X$ be a multi-set of integers, $b$ be a positive integer greater than $1$, and $\mathcal{C}$ be a collection of subsets of $X$. We say $\mathcal{C}$ is a $b$-\textbf{canonical subset collection} of $X$ if every subset in $\mathcal{C}$ has size at most $b$, and for any subset $S\subseteq X$, there exists a subset $\hat{S}\subseteq X$ with $|\hat{S}|=|S|$, $\Sigma(\hat{S})=\Sigma(S)$, and $\hat{S}$ can be expressed as the union of $O((\frac{|S|}{b}+1)\log\frac{1}{\eps})$ disjoint subsets from $\mathcal{C}$.
\end{definition}

Once a $b$-canonical subset collection of $X$ has been computed, Lemma~\ref{lem:subsetsum-fptas-partition-basic-lemma} can be used to ``evenly'' partition $X$ into two subsets. By iterating this procedure, i.e., partitioning $X$ into $X_1, X_2$, then partitioning $X_1$ into $X_{11}, X_{12}$, and so on, we obtain the partition required by Lemma~\ref{lem:subsetsum-fptas-det-partition}. During the iteration, to combine with the dense case, we need to simultaneously handle all subsets produced by the current partition. Let the current partition be $X_1,\ldots,X_{\ell}$. 
The following lemma shows that we can compute $b$-canonical subset collections either for all subsets in the current partition, or for a specific subcollection whose sumset contains a desired arithmetic progression.

\begin{restatable}{lemma}{lemsubsetsumfptasformalsubset}
\label{lem:subsetsum-fptas-formal-subset}
    Let $X$ be a multi-set of $n$ integers from $[\frac{1}{\eps},\frac{2}{\eps})$, $b$ be a positive integer greater than $1$, and $\kappa$ be an integer in $[2,\frac{1}{\eps}]$, with $\Sigma(X)\geq \frac{4\kappa}{\eps}$. Given a partition $X_1,\ldots,X_\ell$ of $X$ ($\ell\leq 2\kappa$), in $\widetilde{O}(b^{O(1)}(n+\ell+\frac{1}{\eps}))$ time, we can  either
    \begin{enumerate}[label={\normalfont(\roman*)}]
        \item find an index set $I\subseteq [1,\ell]$ such that $\sum_{i\in I}\Sigma(X_i)\leq \frac{1}{4}\Sigma(X)$ and $\sum_{i\in I}\mathcal{S}(X_i)$ contains an arithmetic progression $a_1 < \cdots < a_{u'}$ with $u' \geq \frac{2}{\eps}$ and $a_1 < a_{u'} \leq \frac{\kappa}{\eps}$, and for each $i\in I$, compute a $b$-canonical subset collection $\mathcal{C}_i$ of $X_i$, such that the total size $\sum_{i\in I}|\mathcal{C}_i|=\widetilde{O}(bn+b^2\ell + \frac{b^6}{\eps})$; or
        \item for each $i\in[1,\ell]$, compute a $b$-canonical subset collection $\mathcal{C}_i$ of $X_i$, such that the total size $\sum_{i=1}^{\ell}|\mathcal{C}_i|=\widetilde{O}(bn+b^2\ell + \frac{b^6}{\eps})$.
    \end{enumerate}
\end{restatable}

For continuity of the proof, we defer the proof of this lemma to Section~\ref{subsec:subsetsum-fptas-formal-subset-collection}. Assuming Lemma~\ref{lem:subsetsum-fptas-formal-subset} holds, we now explain how to use it to partition the subsets. This proof is similar to the proof of Lemma~2.5 in \cite{Chan26}.

\begin{lemma}\label{lem:subsetsum-fptas-partition-or-dense}
    Let $X$ be a multi-set of integers from $[\frac{1}{\eps},\frac{2}{\eps})$, $k$ be a positive integer that $k\geq c_1\log^{c_2}\frac{1}{\eps}$, and $\kappa$ be an integer in $[2,\frac{1}{\eps}]$, with $\Sigma(X)\geq \frac{4\kappa}{\eps}$. Given a partition $X_1,\ldots,X_\ell$ of $X$ ($\ell\leq 2\kappa$), in $\widetilde{O}(n+\ell+\frac{1}{\eps})$ time, we can  either
    \begin{enumerate}[label={\normalfont(\roman*)}]
        \item find an index set $I\subseteq [1,\ell]$ such that $\sum_{i\in I}\Sigma(X_i)\leq \frac{1}{4}\Sigma(X)$ and $\sum_{i\in I}\mathcal{S}(X_i)$ contains an arithmetic progression $a_1 < \cdots < a_{u'}$ with $u' \geq \frac{2}{\eps}$ and $a_1 < a_{u'} \leq \frac{\kappa}{\eps}$, and for each $i\in I$, compute a partition of $X_i$ into $X_i'$ and $X_i''$ satisfying
        \[
        \mathcal{S}_{\leq k} (X_i) \subseteq \mathcal{S}_{\leq k/2+k/\log k} (X_i')+\mathcal{S}_{\leq k/2+k/\log k} (X_i'');
        \] or
        \item for each $i\in[1,\ell]$, compute a partition of $X_i$ into $X_i'$ and $X_i''$ satisfying
        \[
        \mathcal{S}_{\leq k} (X_i) \subseteq \mathcal{S}_{\leq k/2+k/\log k} (X_i')+\mathcal{S}_{\leq k/2+k/\log k} (X_i'').
        \]
    \end{enumerate}
\end{lemma}

\begin{proof}
    Let $b=\log^{c_2} (\frac{k}{\eps})$. Recall that $c_1,c_2$ are sufficiently large constants. Then we have $k>b$.

    We first use Lemma~\ref{lem:subsetsum-fptas-formal-subset} with parameter $b$ to compute the $b$-canonical subset collections. The running time is $\widetilde{O}(n+\ell+\frac{1}{\eps})$. Depending on the outcome of Lemma~\ref{lem:subsetsum-fptas-formal-subset}, we obtain the collections $\mathcal{C}_i$ either for an index set $I$ with the desired dense properties (corresponding to case (i)), or for all $i \in [1,\ell]$ (corresponding to case (ii)). Let $\mathcal{I}$ denote the set of indices for which $\mathcal{C}_i$ was computed (i.e., $\mathcal{I} = I$ in case (i) and $\mathcal{I} = [1,\ell]$ in case (ii)).

    For each $i \in \mathcal{I}$, we apply Lemma~\ref{lem:subsetsum-fptas-partition-basic-lemma} to compute a partition of $X_i$ into $X_i'$ and $X_i''$ in $\widetilde{O}(b|\mathcal{C}_i|+|X_i|)$ time. This partition guarantees that for any subset $C \in \mathcal{C}_i$, $|C \cap X_i'|, |C \cap X_i''| \leq |C|/2+O(\sqrt{|C|\log |\mathcal{C}_i|})$. 
    
    By the definition of a $b$-canonical subset collection, for any sum $s\in \mathcal{S}_{\leq k}(X_i)$, there exists a valid subset $\hat{S} \subseteq X_i$ with $|\hat{S}| \leq k$ and $\Sigma(\hat{S})=s$ such that $\hat{S}$ can be expressed as a disjoint union of $m = \widetilde{O}(k/b)$ subsets $C_1, \ldots, C_m$ chosen from $\mathcal{C}_i$. We can bound the size of $\hat{S} \cap X_i'$:
    \[
        |\hat{S}\cap X_i'| 
        = \sum_{j=1}^m |C_j \cap X_i'| 
        \leq \sum_{j=1}^m \left( \frac{|C_j|}{2} + O\left(\sqrt{|C_j|\log |\mathcal{C}_i|}\right) \right).
    \]
    Since $\sum_{j=1}^m |C_j| = |\hat{S}|$, and $|C_j| \leq b$ for all $j$, we have  $|\hat{S}\cap X_i'|\leq  |\hat{S}|/2+\widetilde{O}(m\cdot \sqrt{b}) = |\hat{S}|/2+\widetilde{O}(k/\sqrt{b})$. Since $c_2$ is sufficiently large, we ensure that $\widetilde{O}(k/\sqrt{b}) \leq k/\log k$. 

    The same bound holds for $X_i''$. Therefore, we have $|\hat{S}\cap X_i'|, |\hat{S}\cap X_i''| \leq k/2+k/\log k$, which implies $s = \Sigma(\hat{S}) \in \mathcal{S}_{\leq k/2+k/\log k} (X_i')+\mathcal{S}_{\leq k/2+k/\log k} (X_i'')$.

    The total time required by Lemma~\ref{lem:subsetsum-fptas-partition-basic-lemma} is $\widetilde{O}(\sum_i|\mathcal{C}_i|+\sum_i|X_i|)=\widetilde{O}(n+\ell+\frac{1}{\eps})$.
\end{proof}

Our goal is to compute $\mathcal{S}_{\leq 2\kappa}(X)$. We first show that after at most $\lceil\log \kappa \rceil$ rounds of partitioning, we only need to consider the case where each subset contributes $c_1\log^{c_2} \kappa$ elements.

\begin{lemma}\label{lem:subsetsum-fptas-partition-times}
Given an integer $\kappa\in[2,\frac{1}{\eps}]$. If a sequence $k_0, k_1,k_2,\ldots, k_{\lceil\log \kappa \rceil}$ satisfies $k_0=2\kappa$ and $k_i=k_{i-1}/2+k_{i-1}/\log k_{i-1}$ for all $1\leq i\leq \lceil\log \kappa \rceil$, then $k_{\lceil\log \kappa \rceil}=128\log^2 \kappa$.
\end{lemma}

\begin{proof}
The case $\kappa<3$ is trivial. Assume $\kappa \geq 3$. 
Let $t=\lceil\log\kappa\rceil$. It is easy to see that $k_i\geq 1$ for all $i\in [1,t]$.
Rewriting the recurrence as
\[
  k_i = \frac{k_{i-1}}{2}\!\left(1+\frac{2}{\log k_{i-1}}\right).
\]
Expanding iteratively, we have
\begin{equation}\label{eq:ki-expand}
  k_t = \frac{k_0}{2^{t}}\prod_{i=0}^{t-1}\!\left(1+\frac{2}{\log k_i}\right)
      = \frac{2\kappa}{2^{t}}\prod_{i=0}^{t-1}\!\left(1+\frac{2}{\log k_i}\right).
\end{equation}

Since $2^{\lceil\log\kappa\rceil}\geq 2^{\log\kappa}=\kappa$, we have
\[
  \frac{2\kappa}{2^{t}} \leq \frac{2\kappa}{\kappa} = 2.
\]
Next, we estimate the product $\prod_{i=0}^{t-1}(1+2/\log k_i)$.
We first establish a lower bound on $\log k_i$.
Since for any $i\in[1,t]$, $k_i > k_{i-1}/2$ (as $2/\log k_{i-1}>0$), by induction we have
\[
  k_i > \frac{k_0}{2^{i}} = \frac{2\kappa}{2^{i}},
\]
so $\log k_i > \log\kappa - i + 1$. Therefore,
\[
  \sum_{i=0}^{t-1}\frac{2}{\log k_i}
  < \sum_{i=0}^{t-1}\frac{2}{\log\kappa - i + 1}
  \leq \sum_{i=0}^{t-1}\frac{2}{t - i}
  = \sum_{j=1}^{t}\frac{2}{j},
\]
where the second step uses $t=\lceil\log\kappa\rceil\leq \log\kappa+1$.
By the properties of the harmonic series,
\[
  \sum_{j=1}^{t}\frac{2}{j}
  \leq 2\ln\log\kappa + 4.
\]
Using $1+x\leq e^x$ (for all $x\geq 0$), we have
\[
  \prod_{i=0}^{t-1}\!\left(1+\frac{2}{\log k_i}\right)
  \leq \exp\!\left(\sum_{i=0}^{t-1}\frac{2}{\log k_i}\right)
  \leq \exp(2\ln\log\kappa + 4)
  \leq 64\log^2\kappa.
\]

Combining the above two parts, we obtain
\[
  k_{\lceil\log\kappa\rceil}
  \leq 2\cdot 8\log^2\kappa
  = 128\log^2\kappa. \qedhere
\]
\end{proof}

We can now prove Lemma~\ref{lem:subsetsum-fptas-det-partition}.

\lemsubsetsumfptasdetpartition*

\begin{proof}
    If $|X| \leq 4\kappa$, we simply place each element into its own singleton subset (and add empty sets if necessary to reach exactly $4\kappa$ subsets). Since the size of each subset is at most 1, condition (ii) is trivially satisfied. Thus, we assume $|X| > 4\kappa$ for the remainder of the proof.

    We proceed by iteratively applying Lemma~\ref{lem:subsetsum-fptas-partition-or-dense}. Initially, let $X^{(0)} = \{X\}$ and $\ell_0 = 1$. We perform at most $\lceil \log \kappa \rceil$ rounds of partitioning. Let $k_0, k_1, \ldots, k_{\lceil\log \kappa \rceil}$ be the sequence of cardinality bounds defined by $k_0 = 2\kappa$ and $k_r = k_{r-1}/2 + k_{r-1}/\log k_{r-1}$ for all $1 \leq r \leq \lceil\log \kappa \rceil$.  Suppose $r^*$ is the smallest one such that $k_{r^*}< c_1\log^{c_2}\frac{1}{\eps}$.

    In round $r$ (for $r = 1, 2, \ldots, r^*$), let the current partition be $X^{(r-1)} = \{X_1^{(r-1)}, \ldots, X_{\ell_{r-1}}^{(r-1)}\}$. We apply Lemma~\ref{lem:subsetsum-fptas-partition-or-dense} with the parameter $k = k_{r-1}$ to this partition.

    If Lemma~\ref{lem:subsetsum-fptas-partition-or-dense} consistently returns case (ii) across all rounds, then for each $j \in [1, \ell_{r-1}]$, the subset $X_j^{(r-1)}$ is successfully partitioned into $X_j'$ and $X_j''$ satisfying
    \[
        \mathcal{S}_{\leq k_{r-1}}(X_j^{(r-1)}) \subseteq \mathcal{S}_{\leq k_r}(X_j') + \mathcal{S}_{\leq k_r}(X_j'').
    \]
    We define the refined partition for the next round as $X^{(r)} = \{X_1', X_1'', \ldots, X_{\ell_{r-1}}', X_{\ell_{r-1}}''\}$, which means the number of subsets doubles: $\ell_r = 2\ell_{r-1} = 2^r$.

    After completing $r^*$ rounds, we obtain a fine partition $X^{(r^*)} = \{X_1, \ldots, X_{2^{r^*}}\}$. By chaining the inclusions from round $1$ to $r^*$, we obtain:
    \[
        \mathcal{S}_{\leq k_0}(X) \subseteq \mathcal{S}_{\leq k_1}(X^{(1)}_1) + \mathcal{S}_{\leq k_1}(X^{(1)}_2) \subseteq \cdots \subseteq \sum_{i=1}^{2^{r^*}} \mathcal{S}_{\leq k_{r^*}}\left(X^{(r^*)}_i\right) \subseteq \sum_{i=1}^{2^{r^*}} \mathcal{S}_{\leq c_1\log^{c_2}\kappa}\left(X^{(r^*)}_i\right).
    \]
    By Lemma~\ref{lem:subsetsum-fptas-partition-times}, since $c_1,c_2$ are sufficiently large constants, $r^*< \lceil \log \kappa \rceil$. The number of subsets at this point is $2^{r^*} \leq 4\kappa$.

    Since the total number of elements is $|X| > 4\kappa$, we can arbitrarily split some of the subsets that contain multiple elements until we have exactly $4\kappa$ subsets. The above property clearly remains. Furthermore, because every element in $X$ is at least $\frac{1}{\eps}$, $\sum_{i=1}^{4\kappa} \max(X_i) \geq 4\kappa \cdot \frac{1}{\eps} = \frac{4\kappa}{\eps}$, satisfying all requirements for case (ii).
    
    If Lemma~\ref{lem:subsetsum-fptas-partition-or-dense} returns case (i) at some round $r$, we identify the index set $I \subseteq [1, \ell_{r-1}]$ and define $X' = \bigcup_{j \in I} X_j^{(r-1)}$. By the lemma's guarantee, $\Sigma(X') \leq \frac{1}{4}\Sigma(X^{(r-1)}) \leq \frac{1}{4}\Sigma(X)$, and $\sum_{j\in I} \mathcal{S}(X_j^{(r-1)})$ contains the desired arithmetic progression (which implies $\mathcal{S}(X')$ also contains it). Furthermore, the lemma provides a partition of $X_j^{(r-1)}$ into $X_j'$ and $X_j''$ for each $j \in I$. At this point, we safely discard all subsets outside of $I$, and set the partition for the next round as $X^{(r)} = \{X'_j, X''_j : j\in I\}$. 

    We then continue the remaining iterations $r+1, \dots, r^*$. Whenever Lemma~\ref{lem:subsetsum-fptas-partition-or-dense} returns case (i) again, we repeat this process by updating $X'$ to be the union of the newly selected subsets and discarding the unselected ones. Since this process only strictly decreases the sum of $X'$ while preserving the existence of the arithmetic progression, after completing all $r^*$ rounds, the final resulting subset $X'$ still satisfies $\Sigma(X') \leq \frac{1}{4}\Sigma(X)$ and $\mathcal{S}(X')$ contains the desired arithmetic progression, and we obtain a partition $X'_1, \dots, X'_{4\kappa}$. Analogous to the analysis in case (ii), we have
    \[
        \mathcal{S}_{\leq 2\kappa}(X') \subseteq \sum_{i=1}^{4\kappa} \mathcal{S}_{\leq c_1\log^{c_2}\kappa}(X'_i).
    \] 

    Finally, we bound the running time. In any round $r$, whether we are processing the full $X$ or the pruned $X'$, the number of subsets is at most $2^{r^*} \leq 2\kappa = O(\kappa)$. Therefore, each invocation of Lemma~\ref{lem:subsetsum-fptas-partition-or-dense} takes $\widetilde{O}(n + \kappa + \frac{1}{\eps})$ time. Over at most $\lceil \log \kappa \rceil$ rounds, the total time complexity evaluates to
    \[
        \widetilde{O}\left(\lceil \log \kappa \rceil \cdot \left(n + \kappa + \frac{1}{\eps}\right)\right) = \widetilde{O}\left(n + \kappa + \frac{1}{\eps}\right). \qedhere
    \]
\end{proof}

\subsection{Computing Canonical Subset Collections}\label{subsec:subsetsum-fptas-formal-subset-collection}
In this subsection, we focus on proving Lemma~\ref{lem:subsetsum-fptas-formal-subset}.
We first illustrate the idea, originally proposed by Chan~\cite{Chan26}, for computing a canonical subset collection for a single set $X$. Build a binary tree on $X$: the root corresponds to $X$, and we recursively split each node's subset into two disjoint children, until every leaf corresponds to a subset of size at most $b$.
Let $X_v \subseteq X$ denote the subset associated with node $v$. For each node $v$, we compute $\mathcal{S}_1(X_v),\ldots,\mathcal{S}_b(X_v)$, recalling that $\mathcal{S}_k(X)$ denote the set of sums of all $k$-element subsets of $X$. For every value in these sets, we find a subset whose sum equals that element. All such subsets form the collection $\mathcal{C}$.

We briefly explain why $\mathcal{C}$ is a $b$-canonical subset collection of $X$. Consider any subset $Y$ of $X$. In the binary tree, we say a node $v$ is selected by $Y$ if $|Y\cap X_v|\leq b$ and $|Y\cap X_{\mathrm{par}(v)}|>b$, where $\mathrm{par}(v)$ denotes the parent of $v$. The number of nodes selected by $Y$ is at most $\widetilde{O}(|Y|/b+1)$. 
Furthermore, the intersections $Y \cap X_v$ over all selected nodes $v$ form a partition of $Y$, and each intersection has size at most $b$.
Therefore, for each intersection $Y\cap X_v$, there exists an element in $\mathcal{C}$ with the same sum and the same size. Taking the union of these subsets yields $\hat{S}$, satisfying $|\hat{S}|=|Y|$ and $\Sigma(\hat{S})=\Sigma(Y)$.

Computing $\mathcal{S}_1(X_v),\ldots,\mathcal{S}_b(X_v)$ is the main computational overhead of this idea. Chan's algorithm is too slow for our purposes, as its running time depends on the maximum element in $X_v$ or the size of $\mathcal{S}(X_v)$. To overcome this, we propose a novel algorithm that runs in $\widetilde{O}(b^{O(1)}\sum_{i=1}^{b^2} |\widehat{\mathcal{S}}_i(X)|)$ time. Since $b$ is set to a polylogarithmic value in our applications, the extra computation does not affect the time complexity. 
Our algorithm utilizes the hash function computations introduced in~\cite[Section~4]{AYZ95}, a technique also adopted by Chan~\cite{Chan26}.

\begin{lemma}\label{lem:subsetsum-fptas-subset-compute}
    Given a multi-set $X$ of $n$ integers and a positive integer $b$, in $\widetilde{O}(n+b^{O(1)}\sum_{i=1}^{b^2} |\widehat{\mathcal{S}}_i(X)|)$ time, we can compute sets $\widehat{\mathcal{S}}_1(X),\ldots,\widehat{\mathcal{S}}_{b^2}(X)$ such that $\widehat{\mathcal{S}}_i(X)=\mathcal{S}_i(X)$ for all $i\leq b$, and $\widehat{\mathcal{S}}_i(X)\subseteq \mathcal{S}_i(X)$ for all $b<i\leq b^2$.

    Moreover, for any $i\in[0,b^2]$ and $y\in \widehat{\mathcal{S}}_i(X)$, we can return a subset $Y\subseteq X$ with $\Sigma(Y)=y$ and $|Y|=i$ in $O(i)$ time.
\end{lemma}

\begin{proof}
Let $n=|X|$.  There exists a ``perfect'' hash family $\mathcal{H}$ of size $O(b^{O(1)}\log n)$, where each function maps the indexed elements of $X$ to $\{1,\ldots,b^2\}$ and satisfies the following property: for every subset $S\subseteq X$ with $|S|\leq b$, there exists $h\in\mathcal{H}$ such that the values $\{h(x):x\in S\}$ are all distinct.  Such a family can be constructed in $O(n+b^{O(1)}\log^{O(1)}n)$ time; see~\cite[Section~4]{AYZ95}.

Fix $h\in\mathcal{H}$.  For each $j\in[1,b^2]$, define
\[
    X_{h,j}:=\{x\in X:h(x)=j\}.
\]
Thus $X_{h,1},\ldots,X_{h,b^2}$ form a partition of $X$. We now describe the dynamic program used for this fixed hash function $h$.

Instead of storing only the set of all reachable sums, we maintain one set for each possible cardinality.  After processing the first $j$ hash subsets, the set $P_r^{(j)}$ stores all sums that can be formed by choosing exactly $r$ elements from the first $j$ subsets, using at most one element from each subset.

Formally, for $j\in[0,b^2]$ and $r\in[0,b^2]$, define
\[
    P_r^{(j)}
    :=
    \left\{
      \sum_{x\in Y}x:
      Y\subseteq X_{h,1}\cup\cdots\cup X_{h,j},\
      |Y|=r,\ 
      |Y\cap X_{h,q}|\leq 1 \text{ for all }q\leq j
    \right\}.
\]
The initialization is
\[
    P_0^{(0)}=\{0\},\qquad
    P_r^{(0)}=\emptyset \quad\text{for all }r\geq 1.
\]
When subset $X_{h,j}$ is processed, each element in it is either chosen or not chosen. Hence the transition is
\[
    P_r^{(j)}
    =
    P_r^{(j-1)}
    \cup
    \bigl(P_{r-1}^{(j-1)}+X_{h,j}\bigr),
    \qquad r=1,\ldots,b^2,
\]
with $P_0^{(j)}=\{0\}$ for all $j\geq 1$. 

Based on the defined transition, we compute the sets $P_r^{(j)}$ iteratively for $j=1,\ldots,b^2$. In each step $j$, the sumset $P_{r-1}^{(j-1)}+X_{h,j}$ is evaluated utilizing Lemma~\ref{lem:subsetsum-fptas-sumset-compute}. To enable solution reconstruction, for any new sum $z \in P_r^{(j)} \setminus P_r^{(j-1)}$ formed by $z = y + x$ (where $y \in P_{r-1}^{(j-1)}$ and $x \in X_{h,j}$), we record a witness pointer $\mathrm{wit}_h(z,r) := (y, r-1)$. 

After processing all $b^2$ subsets, we set $S_r^h := P_r^{(b^2)}$ for $r=1,\ldots,b^2$. Applying this dynamic programming procedure for all hash functions $h \in \mathcal{H}$, we define
\[
    \widehat{\mathcal{S}}_i(X)
    :=
    \bigcup_{h\in\mathcal{H}} S_i^h,
    \qquad i=1,\ldots,b^2.
\]

\paragraph{Correctness.}
We first prove the correctness of the DP. 
We now prove separately:
(i)~$\widehat{\mathcal{S}}_i(X)\subseteq \mathcal{S}_i(X)$ for all $i\leq b^2$;
(ii)~$\mathcal{S}_i(X)\subseteq \widehat{\mathcal{S}}_i(X)$ for all $i\leq b$.

First, for every $i\leq b^2$ and every $h\in\mathcal{H}$, each $y\in S_i^h=P_i^{(b^2)}$ is the sum of exactly $i$ elements of $X$, using at most one element from each hash subset. Therefore
\[
    S_i^h\subseteq \mathcal{S}_i(X)
\]
for every $h$, and consequently
\[
    \widehat{\mathcal{S}}_i(X)
    =
    \bigcup_{h\in\mathcal{H}}S_i^h
    \subseteq
    \mathcal{S}_i(X)
\]
for every $i\leq b^2$.

Second, let $i\leq b$ and let $y\in\mathcal{S}_i(X)$.  Then there exists a subset $Y=\{x_1,\ldots,x_i\}\subseteq X$ with $\Sigma(Y)=y$.  By the defining property of the perfect hash family, there is a function $h\in\mathcal{H}$ that that $h(x_1),\ldots,h(x_i)$ are all distinct. Therefore the elements of $Y$ lie in $i$ distinct subsets
$X_{h,j_1},\ldots,X_{h,j_i}$, where we may assume
$j_1<\cdots<j_i$.

By the DP, after subset $X_{h,j_q}$ has been processed, the partial sum $x_1+\cdots+x_q$ belongs to set $P_q^{(j_q)}$.  Inducting on $q=1,\ldots,i$ shows that after all subsets are processed, the full sum
\[
    y=x_1+\cdots+x_i
\]
belongs to $S_i^h$.  Hence $y\in\widehat{\mathcal{S}}_i(X)$.  We have proved
\[
    \mathcal{S}_i(X)\subseteq \widehat{\mathcal{S}}_i(X)
    \qquad\text{for every }i\leq b.
\]

Combining this with the previous inclusion gives
$\widehat{\mathcal{S}}_i(X)=\mathcal{S}_i(X)$ for all $i\leq b$, and
$\widehat{\mathcal{S}}_i(X)\subseteq\mathcal{S}_i(X)$ for all
$b<i\leq b^2$.

\paragraph{Witness recovery.}
Whenever a new value $z$ is inserted into set $P_r^{(j)}$ through the transition
$z=y+x$, where $y\in P_{r-1}^{(j-1)}$ and $x\in X_{h,j}$, we store
\[
    \mathrm{wit}_h(z,r):=(y,r-1).
\]
To answer a query $(y,i)$ with $y\in\widehat{\mathcal{S}}_i(X)$, first choose
an index $h$ such that $y\in S_i^h$.  Starting from $(y,i)$, repeatedly follow
the witness pointer
\[
    (y,i)\mapsto (y',i-1),
\]
outputting the element $y-y'$ at each step and continuing from $(y',i-1)$. After exactly $i$ steps, the process reaches $0$, whose only element is $0$.  The recovered elements form a subset $Y\subseteq X$ with $|Y|=i$ and $\Sigma(Y)=y$.  Thus witness recovery takes $O(i)$ time.

\paragraph{Running Time Analysis.}
Fix $h\in\mathcal{H}$.  During the DP, for every subset index $j$ and every cardinality index $r$, the algorithm computes one sumset
\[
    P_{r-1}^{(j-1)}+X_{h,j}.
\]
By Lemma~\ref{lem:subsetsum-fptas-sumset-compute}, this costs
    $\widetilde{O}\bigl(|P_{r-1}^{(j-1)}+X_{h,j}|\bigr)$
time, and it also provides one representation $z=y+x$ for every output value $z$.  Since every newly inserted value is placed into some set $P_r$, and there are $b^2$ subsets and the cardinality is at most $b^2$, the total work for this fixed
$h$ is bounded by
\[
    \widetilde{O}\!\left(
      b^{O(1)}\sum_{i=1}^{b^2}|S_i^h|
    \right).
\]
Summing over all $h\in\mathcal{H}$ gives
\[
    \sum_{h\in\mathcal{H}}
    \widetilde{O}\!\left(
      b^{O(1)}\sum_{i=1}^{b^2}|S_i^h|
    \right)
    \leq
    \widetilde{O}\!\left(
      b^{O(1)}\sum_{i=1}^{b^2}|\widehat{\mathcal{S}}_i(X)|
    \right),
\]
where the size of the hash family $|\mathcal{H}|=O(b^{O(1)}\log n)$. Including the construction time of the hash family, the total running time is
\[
    \widetilde{O}\!\left(n+
      b^{O(1)}
      \sum_{i=1}^{b^2}|\widehat{\mathcal{S}}_i(X)|
    \right).\qedhere
\]
\end{proof}

We can now prove Lemma~\ref{lem:subsetsum-fptas-formal-subset}. To combine this with our density analysis, we must handle all generated subsets simultaneously. We construct a tree for each subset, and compute their bounded-cardinality sumsets level by level. At each level, whenever the total size of the computed sets becomes sufficiently large, we stop the computation for the remaining nodes and move to the next level. After computing all levels, for the subsets that were fully computed, we extract their $b$-canonical subset collections and show that the sumsets contains the desired arithmetic progression (if stop early).
After computing all levels, we extract $b$-canonical subset collections for the subsets that were fully computed. If early stopping was triggered, we show that the sumsets of the computed sets contain the desired arithmetic progression.

\lemsubsetsumfptasformalsubset*

\begin{proof}
    For each $i \in [1, \ell]$, we build a binary tree of depth $L := \lceil\log\frac{1}{\eps}\rceil$ on $X_i$. 
    At level $h$ ($h = 0, 1, \ldots, L$), we define
    \[
        X_{i,j}^{(h)} := X_i \cap \Bigl[\frac{1}{\eps} + \frac{j-1}{2^h\eps}, \; \frac{1}{\eps} + \frac{j}{2^h\eps}\Bigr), \quad j \in [1, 2^h].
    \]
    Note that at level $h$, the collection $\{X_{i,j}^{(h)}\}_{j=1}^{2^h}$ forms a partition of $X_i$, and the elements in each subset $X_{i,j}^{(h)}$ all belong to an interval of diameter $\frac{1}{2^h\eps}$.

    We compute the sumsets level by level. At level $h$, for each nonempty subset $X_{i,j}^{(h)}$, we invoke Lemma~\ref{lem:subsetsum-fptas-subset-compute} to compute $\widehat{\mathcal{S}}_1(X_{i,j}^{(h)}), \ldots, \widehat{\mathcal{S}}_{b^2}(X_{i,j}^{(h)})$. For each $k \in [1, b]$ and each $y \in \widehat{\mathcal{S}}_k(X_{i,j}^{(h)})$, we find a corresponding subset $Y \subseteq X_{i,j}^{(h)}$ with $|Y| = k$ and $\Sigma(Y) = y$, and add $Y$ to $\mathcal{C}_i$.

    Set the threshold
    \[
        \Gamma := b^2\ell + \frac{4Cb^6}{\eps}\log\frac{b^2}{\eps},
    \]
    where $C$ is the constant from Lemma~\ref{lem:ap-with-small-element}. At each level $h$, we process the indices sequentially for $i = 1, \ldots, \ell$, computing all $j \in [1, 2^h]$ for each $i$. If, immediately after processing some index $i_h$, the total size of all sumsets computed at this level reaches or exceeds the threshold, i.e.,
    \[
        \sum_{i=1}^{i_h} \sum_{j=1}^{2^h} \sum_{k=1}^{b^2} |\widehat{\mathcal{S}}_k(X_{i,j}^{(h)})| \geq \Gamma,
    \]
    we abort the computation for the remaining indices $i > i_h$ at this level and immediately proceed to compute the next level $h+1$. If the threshold is never reached at level $h$, we define $i_h = \ell$.

    After completing all $L$ levels, each nonempty leaf node $X_{i,j}^{(L)}$ has an interval diameter of $\frac{1}{2^L\eps} \leq 1$. Since all elements are integers, this implies $|\mathtt{supp}(X_{i,j}^{(L)})| = 1$, meaning all elements in a given leaf subset are identical. We renumber all nonempty leaf subsets $X_{i,j}^{(L)}$ across all indices as $X'_1, \ldots, X'_m$. For each $X'_r$, we split it evenly into $\lceil|X'_r|/b\rceil$ small blocks, each containing at most $b$ elements. For each block $X''$ and each $k \in [1, |X''|]$, we pick an arbitrary subset of size $k$ and add it to the corresponding collection $\mathcal{C}_i$.

    \paragraph{Properties of Canonical Collections.}
    Let $i^* = \min_{h=0}^L i_h$. For every $i \in [1, i^*]$, the tree for $X_i$ is fully computed without any early stopping. We first verify that for these fully processed indices, $\mathcal{C}_i$ satisfies the properties of a $b$-canonical subset collection. Every subset added to $\mathcal{C}_i$ has a size of at most $b$, which follows immediately from the construction.

    For any subset $S \subseteq X_i$, we can greedily search top-down in the tree. At the current node corresponding to $X'$: if $|S \cap X'| \leq b$, we let $b' = |S \cap X'|$. Because $\Sigma(S \cap X') \in \widehat{\mathcal{S}}_{b'}(X')$, there exists a subset $Y \in \mathcal{C}_i$ such that $|Y| = b'$, $\Sigma(Y) = \Sigma(S \cap X')$, and $Y \subseteq X'$. We select this node and stop searching downward. If $|S \cap X'| > b$, we continue recursing into its children.
        
    In the first $L$ levels, since the parent of each selected node contains at least $b$ elements of $S$, at most $2(\lfloor|S|/b\rfloor + 1)$ nodes per level halt the search. Thus, the total number of selected nodes in the internal levels is bounded by $O((\frac{|S|}{b}+1)\log\frac{1}{\eps})$.
    For nodes that reach level $L$ with $|S \cap X'| > b$, all elements in $X'$ are identical. We can cover $S \cap X'$ using the blocks from the final step, utilizing $O(|S \cap X'|/b + 1)$ subsets per node. Over all such leaf nodes, the number of subsets used is at most $O(\frac{|S|}{b}+1)$.
    Let $\hat{S}$ be the union of all these selected subsets. By construction, $|\hat{S}| = |S|$, $\Sigma(\hat{S}) = \Sigma(S)$, and $\hat{S}$ is composed of $O((\frac{|S|}{b}+1)\log\frac{1}{\eps})$ disjoint subsets from $\mathcal{C}_i$. Thus, $\mathcal{C}_i$ is a valid $b$-canonical subset collection.

    We now bound the total size of the computed collections. For any single index $i$, the total number of distinct sumsets generated at level $h$ is bounded by the maximum possible values. Since $\widehat{\mathcal{S}}_k(X_{i,j}^{(h)}) \subseteq [\frac{k}{\eps} + \frac{(j-1)k}{2^h\eps}, \frac{k}{\eps} + \frac{jk}{2^h\eps})$, the sumsets for different $j$ are disjoint. The total size across all $j$ and $k \leq b^2$ is bounded by $\sum_{k=1}^{b^2} O(\frac{k}{\eps}) = O(\frac{b^4}{\eps})$. 
    Therefore, the total number of subsets added to $\mathcal{C}$ at each level is at most $\Gamma + O(\frac{b^4}{\eps}) = O(\Gamma)$. Over $L = O(\log \frac{1}{\eps})$ levels, this contributes $O(\Gamma \log \frac{1}{\eps})$ subsets. The final leaf-partitioning step produces at most $b \cdot \lceil|X'_r|/b\rceil \leq |X'_r| + b$ subsets per block, totaling $O(n + bm) = O(bn)$ subsets. Hence, the total size is:
    \[
        \sum_{i=1}^{i^*} |\mathcal{C}_i| 
        = O\Bigl(\Gamma\log\frac{1}{\eps} + bn\Bigr) 
        = O\Bigl(b^2\ell\log\frac{1}{\eps} + \frac{b^6}{\eps}\log\frac{b}{\eps}\log\frac{1}{\eps} + bn\Bigr) 
        = \widetilde{O}\Bigl(bn + b^2\ell + \frac{b^6}{\eps}\Bigr).
    \]

    If $i^* = \ell$, we have successfully computed $\mathcal{C}_i$ for all $i \in [1, \ell]$, fulfilling condition (ii) of the lemma.

    \paragraph{Early Stopping and Case (i).}
    Otherwise, $i^* < \ell$, meaning early stopping was triggered at some level $h$ where $i_h = i^*$. At this level, we have:
    \[
        \sum_{i=1}^{i^*} \sum_{j=1}^{2^h} \sum_{k=1}^{b^2} |\widehat{\mathcal{S}}_k(X_{i,j}^{(h)})| \geq \Gamma.
    \]
    By the pigeonhole principle, there exists a cardinality $b^* \in [1, b^2]$ such that
    \[
        \sum_{i=1}^{i^*} \sum_{j=1}^{2^h} |\widehat{\mathcal{S}}_{b^*}(X_{i,j}^{(h)})| 
        \geq \frac{\Gamma}{b^2} 
        = \ell + \frac{4Cb^4}{\eps}\log\frac{b^2}{\eps}.
    \]
    For each $i \in [1, i^*]$, let $T_i := \bigcup_{j=1}^{2^h} \widehat{\mathcal{S}}_{b^*}(X_{i,j}^{(h)})$. Because the sub-intervals for different $j$ are pairwise disjoint, $|T_i| = \sum_j |\widehat{\mathcal{S}}_{b^*}(X_{i,j}^{(h)})|$. Furthermore, $T_i \subseteq \mathcal{S}_{b^*}(X_i) \subseteq \mathcal{S}(X_i)$. Thus, the maximum element in $T_i$ satisfies $\max(T_i) \leq \frac{2b^*}{\eps} \leq \frac{2b^2}{\eps}$, and the diameter $\dm(T_i) \leq \frac{b^*}{\eps}\leq \frac{b^2}{\eps}$. 
    
    With $\sum_{i=1}^{i^*} |T_i| \geq \ell + \frac{4Cb^4}{\eps}\log\frac{b^2}{\eps}$, we apply Lemma~\ref{lem:ap-with-small-element} on the sets $T_1, \ldots, T_{i^*}$ (using parameters $u' = \frac{b^2}{\eps}$, $\rho = 4b^2$, and $f(i) = \Sigma(X_i)$) and select an index set $I \subseteq [1, i^*]$ such that $|I| \leq \frac{\ell}{\rho}$ and the sumset $\sum_{i \in I} T_i$ (and thus $\sum_{i \in I} \mathcal{S}(X_i)$) contains an arithmetic progression $a_1 < \cdots < a_{u'}$ of length at least $u' \geq \frac{2}{\eps}$ and $a_1 < a_{u'} \leq |I| \cdot \max(T_i) \leq \frac{\ell}{4b^2} \cdot \frac{2b^2}{\eps} \leq \frac{\ell}{2\eps} \leq \frac{\kappa}{\eps}$. Additionally, the weight function ensures:
    \[
        \sum_{i \in I} \Sigma(X_i) \leq \frac{1}{\rho} \sum_{i=1}^{\ell} \Sigma(X_i) \leq \frac{1}{4} \Sigma(X),
    \]
    satisfying all requirements of condition (i).

    \paragraph{Running Time Analysis.}
    Finally, we analyze the running time of the algorithm. The dominant cost comes from invoking Lemma~\ref{lem:subsetsum-fptas-subset-compute}. For any level $h$, the time spent is bounded by $\widetilde{O}(b^{O(1)}(\Gamma+\frac{b^4}{\eps}))$. Over $O(\log \frac{1}{\eps})$ levels, the overall running time is $\widetilde{O}(b^{O(1)}(\Gamma + \frac{b^4}{\eps}))=\widetilde{O}(b^{O(1)}(n + \ell + \frac{1}{\eps}))$, concluding the proof.
\end{proof}

\section{Base Layer}\label{sec:base-layer}
By Lemma~\ref{lem:subsetsum-fptas-det-partition}, we partition the set into $4\kappa$ subsets, each contributing $c_1\log^{c_2}\kappa$ elements, where $c_1,c_2$ are sufficiently large constants. We first compute, for each subset $X_i$, a set $A_i$ satisfying $\mathcal{S}_{\leq c_1\log^{c_2}\kappa}(X_i)\subseteq A_i\subseteq \mathcal{S}(X_i)$. It then suffices to compute $\sum_i A_i$ to obtain $\mathcal{S}_{\leq 2\kappa}(X)$.

\begin{lemma}\label{lem:subsetsum-fptas-first-level}
    Let $X$ be a multi-set of $n$ integers from $[\frac{1}{\eps}, \frac{2}{\eps})$, $\kappa \in [2, \frac{1}{\eps}]$ be an integer. Given a partition $X_1, \ldots, X_{4\kappa}$ of $X$, in $\widetilde{O}(n + \kappa + \frac{1}{\eps})$ time, we can either:
    \begin{enumerate}[label={\normalfont(\roman*)}]
        \item find an index set $I \subseteq [1, 4\kappa]$ satisfying $\sum_{i \in I} \Sigma(X_i) \leq \frac{1}{4}\Sigma(X)$, and for each $i \in I$, compute a set of integers $A_i \subseteq [0, \frac{c_1^2\log^{2c_2}\kappa}{\eps})$ such that $\mathcal{S}_{\leq c_1\log^{c_2}\kappa}(X_i) \subseteq A_i \subseteq \mathcal{S}(X_i)$, $\sum_{i\in I} |A_i|=\widetilde{O}(\kappa+\frac{1}{\eps})$, and  $\sum_{i \in I} A_i$ contains an arithmetic progression $a_1 < \cdots < a_{u'}$ with $u' \geq \frac{2}{\eps}$ and $a_1 < a_{u'} \leq \frac{\kappa}{\eps}$; or
        \item for each $i \in [1, 4\kappa]$, compute a set of integers $A_i \subseteq [0, \frac{c_1^2\log^{2c_2}\kappa}{\eps})$ such that $\mathcal{S}_{\leq c_1\log^{c_2}\kappa}(X_i) \subseteq A_i \subseteq \mathcal{S}(X_i)$ and $\sum_{i\in I} |A_i|=\widetilde{O}(\kappa+\frac{1}{\eps})$.
    \end{enumerate}
\end{lemma}

\begin{proof}
    Let $b = c_1\log^{c_2}\kappa$. 
    We apply Lemma~\ref{lem:subsetsum-fptas-subset-compute} sequentially to $X_1,\ldots, X_{4\kappa}$, and stop at the first index $i'$ where the cumulative value of $\sum_{i=1}^{i'}\sum_{j=1}^{b^2} |\widehat{\mathcal{S}}_j(X_i)|\geq 4\kappa b^2 + \frac{16Cb^6}{\eps}\log\frac{2b^2}{\eps}$,
    where $C$ is the constant from Lemma~\ref{lem:ap-with-small-element}. 

    For each processed index $i \in [1, i']$, we define $A_i = \bigcup_{j=1}^{b^2} \widehat{\mathcal{S}}_j(X_i) \cup \{0\}$. By Lemma~\ref{lem:subsetsum-fptas-subset-compute}, $A_i \subseteq [0, \frac{c_1^2\log^{2c_2}\kappa}{\eps})$ and $\mathcal{S}_{\leq c_1\log^{c_2}\kappa}(X_i) \subseteq A_i \subseteq \mathcal{S}(X_i)$.

    If the algorithm successfully processes all $4\kappa$ subsets without early stop, we simply output case (ii), as the computed sets $A_1, \ldots, A_{4\kappa}$ satisfy all requisite properties.

    If the algorithm stops early at some index $i' \leq 4\kappa$, we show that it satisifes case (i). Because each $A_i$ is formed by the union of at most $b^2$ sets, we have
    \[
        \sum_{i=1}^{i'} |A_i| \geq \frac{1}{b^2} \left( 4\kappa b^2 +  \frac{16Cb^6}{\eps}\log\frac{2b^2}{\eps} \right) = 4\kappa + \frac{16Cb^4}{\eps}\log\frac{2b^2}{\eps}.
    \]
    We then apply Lemma~\ref{lem:ap-with-small-element} to the sets $A_1, \ldots, A_{i'}$ using parameters $u' = \frac{2b^2}{\eps}$, $\rho = 8b^2$, and $f(i) = \Sigma(X_i)$. We can select an index set $I \subseteq [1, i']$ such that $|I| \leq \frac{i'}{\rho}$ and the sumset $\sum_{i \in I} A_i$ contains an arithmetic progression $a_1 < \cdots < a_{u'}$ of length at least $u' \geq \frac{2}{\eps}$ and $a_1 < a_{u'} \leq |I| \max(A_i)\leq \frac{\kappa}{\eps}$. The weight function implies:
    \[
        \sum_{i \in I} \Sigma(X_i) \leq \frac{1}{\rho} \sum_{i=1}^{i'} \Sigma(X_i) \leq  \frac{1}{4} \Sigma(X).
    \]
    This confirms all the guarantees required for case (i), completing the proof.

    We now analyze the overall running time. Processing a single subset $X_i$ via Lemma~\ref{lem:subsetsum-fptas-subset-compute} requires $O(b^{O(1)} \sum_{j=1}^{b^2} |\widehat{\mathcal{S}}_j(X_i)|)$ time. Since $\widehat{\mathcal{S}}_j(X_i) \subseteq [\frac{j}{\eps}, \frac{2j}{\eps})$, its size is bounded by $O(\frac{j}{\eps}) \le O(\frac{b^2}{\eps})$. Consequently, the final processing step at index $i'$ overshoots the early-stopping threshold by at most $\sum_{j=1}^{b^2} O(\frac{b^2}{\eps}) = O(\frac{b^4}{\eps})$. The overall running time is thus dominated by $\widetilde{O}(n + \kappa + \frac{1}{\eps})$.
\end{proof}

We now combine Lemma~\ref{lem:subsetsum-fptas-partition-or-dense} and Lemma~\ref{lem:subsetsum-fptas-first-level}. In the non-dense case, we directly obtain $\mathcal{S}_{\leq 2\kappa}(X) \subseteq \sum_{i}\mathcal{A}_i \subseteq \mathcal{S}(X)$. In the dense case, we utilize the remaining elements of $X$—those not participating in the arithmetic progression—to extend the progression. This extension yields a sequence that approximates $\mathcal{S}(X)[\frac{\kappa}{\eps},\frac{2\kappa}{\eps}]$ with an additive error of $O(\kappa)$.

\begin{lemma}\label{lem:base-layer}
    Let $X$ be a multi-set of $n$ integers from $[\frac{1}{\eps}, \frac{2}{\eps})$, $\kappa \in [2, \frac{1}{\eps}]$ be an integer, with $\Sigma(X)\geq \frac{4\kappa}{\eps}$. In $\widetilde{O}(n + \kappa + \frac{1}{\eps})$ time, we can compute $8\kappa$ sets $A_1,\ldots,A_{8\kappa}\subseteq [0, \frac{c_1^2\log^{2c_2}\kappa}{\eps})$ such that $\sum_{i=1}^{8\kappa} |A_i|=\widetilde{O}(\kappa+\frac{1}{\eps})$, $\sum_{i=1}^{8\kappa} \max(A_i)\geq \frac{3\kappa}{\eps}$, for any $i\in[1,8\kappa]$, $0\in A_i$, and
    $A_1+\cdots+A_{8\kappa}$ approximates $\mathcal{S}(X)[\frac{\kappa}{\eps},\frac{2\kappa}{\eps}]$ with additive error $O(\kappa)$.
\end{lemma}

\begin{proof}
    We begin by applying Lemma~\ref{lem:subsetsum-fptas-det-partition} to the multiset $X$. This yields a partition of either a subset $X'$ (case i) or the full set $X$ (case ii) into $X_1, \ldots, X_{4\kappa}$. Next, we apply Lemma~\ref{lem:subsetsum-fptas-first-level} to this partition, which also outputs either case (i) or case (ii). We analyze the combined results based on these outcomes:

    \paragraph{Case 1: Both lemmas output case (ii).}
    If both lemmas yield case (ii), Lemma~\ref{lem:subsetsum-fptas-det-partition} guarantees that $X$ is partitioned into $4\kappa$ sets such that $\sum_{i=1}^{4\kappa} \max(X_i) \geq \frac{4\kappa}{\eps}$ and 
    \[
        \mathcal{S}_{\leq 2\kappa}(X) \subseteq \sum_{i=1}^{4\kappa} \mathcal{S}_{\leq c_1\log^{c_2}}(X_i).
    \]
    Lemma~\ref{lem:subsetsum-fptas-first-level} then allows us to compute sets $A_i \subseteq [0, \frac{c_1^2\log^{2c_2}\kappa}{\eps})$ for each $i \in [1, 4\kappa]$ such that $\mathcal{S}_{\leq c_1\log^{c_2}}(X_i) \subseteq A_i \subseteq \mathcal{S}(X_i)$. Combining these relations, we get:
    \[
        \mathcal{S}_{\leq 2\kappa}(X) \subseteq A_1 + \cdots + A_{4\kappa} \subseteq \mathcal{S}(X).
    \]

    Because $\mathcal{S}_{\leq c_1\log^{c_2}}(X_i) \subseteq A_i \subseteq \mathcal{S}(X_i)$, it must be that $\max(A_i) \geq \max(X_i)$. Therefore, $\sum_{i=1}^{4\kappa} \max(A_i) \geq \sum_{i=1}^{4\kappa} \max(X_i) \geq \frac{4\kappa}{\eps}$. 

    To reach $8\kappa$ sets, we add $4\kappa$ trivial sets $A_{4\kappa+1}, \ldots, A_{8\kappa}$ where each is simply $\{0\}$. This ensures:
    \begin{itemize}
        \item $\sum_{i=1}^{8\kappa} |A_i| = \widetilde{O}(\kappa + \frac{1}{\eps})$;
        \item $\sum_{i=1}^{8\kappa} \max(A_i) \geq \frac{3\kappa}{\eps}$;
        \item for any $i\in[1,8\kappa]$, $0\in A_i$;
        \item $\mathcal{S}_{\leq 2\kappa}(X) \subseteq A_1 + \cdots + A_{8\kappa}\subseteq \mathcal{S}(X)$.
    \end{itemize}
    This directly implies that $A_1 + \cdots + A_{8\kappa}$ approximates $\mathcal{S}(X)[\frac{\kappa}{\eps}, \frac{2\kappa}{\eps}]$ with an additive error of $O(\kappa)$.

    \paragraph{Case 2: At least one lemma outputs case (i).}
    If either lemma returns case (i), we identify a subset $X' \subseteq X$ with $\Sigma(X') \leq \frac{1}{4}\Sigma(X)$, partitioned into $\ell \leq 4\kappa$ subsets $X_1, \ldots, X_\ell$. For each $i \in [1, \ell]$, we compute sets $A_i \subseteq [0, \frac{c_1^2\log^{2c_2}\kappa}{\eps})$ satisfying $\mathcal{S}_{\leq c_1\log^{c_2}}(X_i) \subseteq A_i \subseteq \mathcal{S}(X_i)$ and $\sum_{i=1}^\ell |A_i| = \widetilde{O}(\kappa + \frac{1}{\eps})$. 

    $\sum_{i=1}^\ell A_i$ contains an arithmetic progression $a_1 < \cdots < a_{u'}$ where $u' \geq \frac{2}{\eps}$ and $a_1 < a_{u'} \leq \frac{\kappa}{\eps}$. This holds because if Lemma~\ref{lem:subsetsum-fptas-first-level} outputs case (i), the sets $A_i$ are explicitly contain this progression. If instead Lemma~\ref{lem:subsetsum-fptas-det-partition} outputs case (i), it guarantees $\mathcal{S}(X')$ contains the progression such that $ a_{u'} \leq \frac{\kappa}{\eps}$. Since every element in $X'$ is at least $\frac{1}{\eps}$, any sum up to $\frac{\kappa}{\eps}$ can use at most $\kappa$ elements. Thus, the progression lies in $\mathcal{S}_{\leq \kappa}(X') \subseteq \mathcal{S}_{\leq 2\kappa}(X')$. Through the subset relations, $\mathcal{S}_{\leq 2\kappa}(X') \subseteq \sum_{i=1}^{\ell} \mathcal{S}_{\leq c_1\log^{c_2}}(X_i) \subseteq \sum_{i=1}^\ell A_i$, ensuring the progression is always captured. The common difference of this progression is $\Delta \leq \frac{\kappa/\eps}{u' - 1} \leq \kappa$.

    We then construct $8\kappa - \ell$ additional sets $A_{\ell+1}, \ldots, A_{8\kappa}$ of the form $\{0, x_j\}$ for elements $x_j \in X \setminus X'$. If $|X \setminus X'| \leq 8\kappa - \ell$, we use all available elements and add enouth sets where each is $\{0\}$, giving $\sum_{i=\ell+1}^{8\kappa} \max(A_i) = \Sigma(X \setminus X') \geq \frac{3}{4}\Sigma(X) \geq \frac{3\kappa}{\eps}$. Otherwise, we select exactly $8\kappa - \ell$ elements. Because every element is at least $\frac{1}{\eps}$ and $\ell \leq 4\kappa$, we have $\sum_{i=\ell+1}^{8\kappa} \max(A_i) \geq (8\kappa - \ell)\frac{1}{\eps} \geq \frac{4\kappa}{\eps} \geq \frac{3\kappa}{\eps}$.

    In both scenarios, $\sum_{i=1}^{8\kappa} \max(A_i) \geq \frac{3\kappa}{\eps}$. The total size bound is maintained since $\sum_{i=1}^{8\kappa} |A_i| \leq \widetilde{O}(\kappa + \frac{1}{\eps}) + 16\kappa = \widetilde{O}(\kappa + \frac{1}{\eps})$. It is clear that for any $i\in[1,8\kappa]$, $0\in A_i$, and $A_1 + \cdots + A_{8\kappa} \subseteq \mathcal{S}(X)$.

    By Lemma~\ref{lem:ap-extend}, $A_1 + \cdots + A_{8\kappa}$ contains a sequence $s_1 < \cdots < s_k$ such that $s_1 = a_1 \leq \frac{\kappa}{\eps}$ and $s_k = a_{u'} + \sum_{i=\ell+1}^{8\kappa} \max(A_i) \geq \frac{3\kappa}{\eps}$. For any $i \in [2, k]$, $s_i - s_{i-1} \leq \Delta \leq \kappa$. 
    Finally, by Lemma~\ref{lem:subsetsum-fptas-sequence-appro-set}, this sequence ensures that $A_1 + \cdots + A_{8\kappa}$ approximates $\mathcal{S}(X)[\frac{\kappa}{\eps}, \frac{2\kappa}{\eps}]$ with an additive error of $\kappa$.

    \paragraph{Running Time Analysis.}
    The total running time is dominated by the running time of Lemma~\ref{lem:subsetsum-fptas-det-partition} and Lemma~\ref{lem:subsetsum-fptas-first-level}. Both of these lemmas are guaranteed to run in $\widetilde{O}(n + \kappa + \frac{1}{\eps})$ time. The remaining operation in our algorithm is adding sets to exactly $8\kappa$ sets, which takes $O(\kappa)$ time. Therefore, the overall running time is $\widetilde{O}(n + \kappa + \frac{1}{\eps})$.
\end{proof}

\section{Tree-Based Computation}\label{sec:subsetsum-fptas-tree-merge}
Let $\{A_1,\cdots,A_{8\kappa}\}$ be the sets computed in the above section. We now compute $\sum_{i=1}^{8\kappa} A_i$.
Taking $\{A_1,\cdots,A_{8\kappa}\}$ as the bottom level (level $0$) of the tree structure, we merge upward level by level. Consider computing level $h$ of the tree structure. Let $A_1, \ldots, A_\ell$ be the nodes at this level and $B_1, \ldots, B_{\ell/2}$ be the nodes at the next level (level $h+1$). We have $A_1, \ldots, A_\ell \subseteq [0, \frac{c_1^2\log^{2c_2}\kappa\cdot 2^h}{\eps}]$. In the approximation scheme, we can rescale and round then to $\widetilde{O}(\frac{1}{\eps})$, allowing us to combine the computation and the density analysis.

\begin{lemma}\label{lem:subsetsum-fptas-tree-operation}
    Let $X$ be a multi-set of $n$ integers from $[\frac{1}{\eps}, \frac{2}{\eps})$, $\kappa \in [2, \frac{1}{\eps}], h\in [0, \log 8\kappa]$ be integers. Let $\ell = 8\kappa/2^h$. Given $A_1, \ldots, A_\ell \subseteq [0, \frac{c_1^2\log^{2c_2}\kappa\cdot 2^h}{\eps}]$ and $0 \in A_i$ for any $i\in[1,\ell]$. Moreover, $\sum_{i=1}^{\ell} \max(A_i) \geq \frac{3\kappa}{\eps}$.
    
    In $\widetilde{O}(\ell+\frac{1}{\eps}+\Sigma_i|A_i|)$ time, we can compute $B_1, \ldots, B_{\ell/2}\subseteq [0, \frac{c_1^2\log^{2c_2}\kappa\cdot 2^{h+1}}{\eps}]$ with total size $\widetilde{O}(\frac{1}{\eps})$, $0 \in B_i$ for any $i\in[1,\ell/2]$. Moreover, $\sum_{i=1}^{\ell/2} \max(B_i) \geq \frac{3\kappa}{\eps}$ and $B_1+\cdots+B_{\ell/2}$ approximates $(A_1+\cdots+A_{\ell})[\frac{\kappa}{\eps},\frac{2\kappa}{\eps}]$ with additive error $O(\kappa\log^{2c_2}\kappa)$.
\end{lemma}

\begin{proof}
    We can round each element of $A_i$ up to the nearest multiple of $c_1^2\log^{2c_2}\kappa\cdot 2^h$. This incurs an additive error of ${O}(\log^4\kappa\cdot 2^h)$ for each $A_i$. For $i \in [1, \ell]$, define $A'_i = \{\frac{a}{c_1^2\log^{2c_2}\kappa\cdot 2^h} : a \in A_i\}$. Then $A'_1, \ldots, A'_\ell$ are subsets of $[0, \frac{1}{\eps}]$. For $i \in [1, \ell/2]$, define $B'_{i} = A'_{2i-1} + A'_{2i}$.

    We compute the $B'_i$'s sequentially via Lemma~\ref{lem:subsetsum-fptas-sumset-compute}. We stop when some $i'$ satisfies $\sum_{i=1}^{i'}|B'_i|\geq i'+\frac{32 C c_1^2\log^{2c_2} \kappa}{\eps}\log\frac{4}{\eps}$. For any $i\in[1,\ell/2]$, computing $B'_i$ takes $\widetilde{O}(|B'_i|)$ time. Since $B'_i\in[0,\frac{2}{\eps}]$, $|B_i'|=O(\frac{1}{\eps})$. Therefore, the total running time is $\widetilde{O}(\ell+\frac{1}{\eps}+\Sigma_i|A_i|)$.

    If there is no early stop, we can clearly compute $B'_1, \ldots, B'_{\ell/2}$ with total size $O(\frac{1}{\eps} \log^4\kappa)$. Let $B_i=B'_i\cdot c_1^2\log^{2c_2}\kappa\cdot 2^h$ for each $i \in [1, \ell/2]$. Then $B_i$ approximates $A_{2\ell-1}+A_{2\ell}$ with additive error ${O}(\log^4\kappa\cdot 2^h)$. So $B_1+\cdots+B_{\ell/2}$ approximates $A_1+\cdots+A_{\ell}$ with additive error $O(\kappa\log^4\kappa)$.

    If it stop early, we apply Lemma~\ref{lem:ap-with-small-element} (with $\rho=8c_1^2\log^{2c_2}\kappa$, $u'=\frac{4}{\eps}$, $f(i)=\max(B'_i)$) to the computed sets $\{B'_1,\ldots,B'_{i'}\}$ and select an index set $I \subseteq [1, i']$ such that $|I| \leq \frac{\ell}{2\rho}$ and the sumset $\sum_{i \in I} B'_i$ contains an arithmetic progression $a'_1 < \cdots < a'_{u'}$ of length at least $u' \geq \frac{2}{\eps}$ and $a'_1 < a'_{u'} \leq |I| \cdot \max(B'_i) \leq \frac{\ell}{2\rho} \cdot \frac{2}{\eps} \leq \frac{\ell}{8c_1^2\log^{2c_2}\kappa\cdot\eps}$. Additionally, the weight function ensures:
    \[
        \sum_{i \in I} \max(B'_i) \leq \frac{1}{\rho} \sum_{i=1}^{i'} \max(B'_i) \leq \frac{1}{4} \sum_{i=1}^{\ell} \max(A'_i).
    \]

    For each $i\in [1,\ell/2]\backslash I$, let $B'_i = \{0,\max(A'_{2i-1})+\max(A'_{2i})\}$. we have
    \[
        \sum_{i \notin I} \max(B'_i) \geq \frac{3}{4} \sum_{i=1}^{\ell} \max(A'_i)\geq \frac{2\kappa}{c_1^2\log^{2c_2}\kappa\cdot 2^h\eps}.
    \]
    By Lemma~\ref{lem:ap-extend}, $B'_1+\cdots+B'_{\ell/2}$ contains a sequence $s'_1<\cdots<s'_k$ such that $s'_1=a'_1\leq \frac{\ell}{8c_1^2\log^{2c_2}\kappa\cdot\eps}$, $s'_k=a_{u'}+\sum_{i \notin I} \max(B'_i)\geq \frac{2\kappa}{c_1^2\log^{2c_2}\kappa\cdot 2^h\eps}$ and $s_i-s_{i-1}\leq \frac{\kappa}{c_1^2\log^{2c_2}\kappa\cdot 2^h}$ for $i\in [2,k]$. It is easy to see that $B'_1+\cdots+B'_{\ell/2}\subseteq A'_1+\cdots+A'_{\ell}$. By Lemma~\ref{lem:subsetsum-fptas-sequence-appro-set}, $B'_1+\cdots+B'_{\ell/2}$ approximates $(A'_1+\cdots+A'_{\ell})[\frac{\kappa}{c_1^2\log^{2c_2}\kappa\cdot 2^h\eps},\frac{2\kappa}{c_1^2\log^{2c_2}\kappa\cdot 2^h\eps}]$ with additive error $\frac{\kappa}{c_1^2\log^{2c_2}\kappa\cdot 2^h}$.

    For each $i \in [1,\ell/2]$, let $B_i=B'_i\cdot c_1^2\log^{2c_2}\kappa\cdot 2^h$. Since $(A'_1+\cdots+A'_{\ell})\cdot c_1^2\log^{2c_2}\kappa\cdot 2^h$ approximate $A_1+\cdots+A_{\ell}$ with additive error $O(\kappa\log^4\kappa)$, $B_1+\cdots+B_{\ell/2}$ approximates $(A_1+\cdots+A_{\ell})[\frac{\kappa}{\eps},\frac{2\kappa}{\eps}]$ with additive error $O(\kappa\log^4\kappa)+\kappa=O(\kappa\log^4\kappa)$.  
    
    It easy to check that $\sum_{i=1}^{\ell/2}|B_i|=\widetilde{O}(\frac{1}{\eps})$, $0 \in B_i$ for any $i\in[1,\ell/2]$ and $\sum_{i=1}^{\ell/2} \max(B_i) \geq \sum_{i=1}^{\ell} \max(A_i)\geq \frac{3\kappa}{\eps}$.
\end{proof}

We can now present the complete proof of the algorithm.
\begin{lemma}
    For any $\kappa\in[2,\frac{1}{\eps}]$, the problem $\mathrm{RP(\kappa)}$ can be solved in $\widetilde{O}(n+\frac{1}{\eps})$ time.
\end{lemma}

\begin{proof}
    We first use Lemma~\ref{lem:base-layer} to compute $8\kappa$ sets $A_1,\ldots,A_{8\kappa}\subseteq [0, \frac{c_1^2\log^{2c_2}\kappa}{\eps})$ such that 
    $A_1+\cdots+A_{8\kappa}$ approximates $\mathcal{S}(X)[\frac{\kappa}{\eps},\frac{2\kappa}{\eps}]$ with additive error $O(\kappa)$.

    We then merge level by level via Lemma~\ref{lem:subsetsum-fptas-tree-operation}, ultimately obtaining a set that approximates $A_1+\ldots+A_{8\kappa}$. In each level, the input size is $\sum_{i=1}^{\ell}|A_i|=\widetilde{O}(\ell+\frac{1}{\eps})$, so the running time is $\widetilde{O}(n+\kappa+\frac{1}{\eps})$. The error caused in one level is $O(\kappa\log^4\kappa)$. The tree structure has $O(\log \kappa)$ levels. So  the total additive error due is $O(\kappa\log^5\kappa)$, and the total running time is $\widetilde{O}(n+\kappa+\frac{1}{\eps})$.
\end{proof}

\section{Conclusion}\label{sec:conclude}

We present a near-linear time weak approximation scheme for Subset Sum in this paper. It would be interesting to explore whether our technique can lead to improvements to other related problems. In particular, it is a major open problem whether there exists an $O(n+x_{\max})$-time exact algorithm for Subset Sum, where $x_{\max}$ refers to the largest input integer. The best known algorithm has a running time $O(n+x_{\max}^{3/2})$~\cite{CLMZ24aSODA}. It would be interesting to close or narrow the gap.   

\section*{Acknowledgements}
We would like to thank Timothy Chan for sharing his paper~\cite{Chan26} with us before it was publicly accessible.

\bibliographystyle{alphaurl}
\bibliography{main}
\end{document}